\documentclass[11pt,letterpaper]{amsart}
\usepackage{header}
\numberwithin{equation}{section}
\begin{document}
\title[Scattering for Schr\"{o}dinger operators with conical decay]{Scattering for Schr\"{o}dinger operators with conical decay}
\author{Adam Black}
\address{Department of Mathematics\\
  Yale University\\
 New Haven, CT 06511}
\email[]{adam.black@yale.edu}
\author{Tal Malinovitch}
\address{Department of Mathematics\\
  Yale University\\
 New Haven, CT 06511}
  \email[]{tal.malinovitch@yale.edu}
\maketitle
\begin{abstract}
     We study the scattering properties of Schr\"{o}dinger operators with potentials that have short-range decay along a collection of rays in $\bbR^d$. This generalizes the classical setting of short-range scattering in which the potential is assumed to decay along \emph{all} rays. For these operators, we show that any state decomposes into an asymptotically free piece and a piece which may interact with the potential for long times. We give a microlocal characterization of the scattering states in terms of the dynamics and a corresponding description of their complement. We also show that in certain cases these characterizations can be purely spatial.
\end{abstract}
\tableofcontents
 \clearpage
\section{Introduction}\label{Intro}
In this paper, we study the scattering properties of Schr\"{o}dinger operators with potentials that have short-range decay along a collection of rays in $\bbR^d$. This generalizes the classical setting of short-range scattering in which the potential is assumed to decay along \emph{all} rays. The classes of potentials we consider will typically be concentrated near some subset of $\bbR^d$ with unbounded complement. For these operators, we show that any state decomposes into an asymptotically free piece and a piece which may interact with the potential for long times. We give a microlocal characterization of the scattering states in terms of the dynamics and a corresponding description of their complement. We also show that in certain cases these characterizations can be purely spatial.\par
Let us now recall the classical picture of short-range scattering in order to situate our result. We consider a self-adjoint Schr\"{o}dinger operator of the form
\begin{align*}
    H=H_0+V
\end{align*}
on $\calH=L^2(\bbR^d)$ where $H_0=-\frac{1}{2}\Delta$ and $V$ is a bounded multiplication operator. There are a variety of decay conditions one can impose on $V$ in order to consider it short-range (see \cite{agmon1975spectral} and the references therein), but we focus on the \emph{Enss condition}
 \begin{align}\label{classicalEnssCondition}
     \|\chi_{B_r^c}V\|_{\textrm{op}}\in L^1([0,\infty),dr)
\end{align}
where $\chi$ is the indicator function of a subset of $\bbR^d$ and $B_r$ is the ball of radius $r$ in $\bbR^d$. This condition was originally posited in \cite{Enss}, in which it was proven that the wave operators
\begin{align*}
    \Omega^\pm=\slim_{t\rightarrow \mp \infty}e^{itH}e^{-itH_0}
\end{align*}
whose range consists of the scattering states, exist on all of $\calH$ and are asymptotically complete. This means that 
\begin{align*}
   \calH=\Ran(\Omega^\pm)\oplus \calH_{\textrm{pp}}(H)
\end{align*}
or equivalently
\begin{align}\label{asymptoticCompleteness}
    \calH_{\textrm{c}}(H)=\Ran(\Omega^\pm)
\end{align}
The proof of this result due to Enss \cite{Enss}, as well as its refinement by Davies \cite{davies1977scattering}, rely on studying the phase space localization of a state as it evolves under $H$. Expanding on these ideas, Kitada and Yajima \cite{kitadaYajima} proved a microlocal characterization of the set of scattering states for time-dependent short-range potentials, among other examples.\par
Motivated by this classical theory, in \cite{black2022scattering} we studied the scattering properties of potentials which are assumed to decay only in some coordinate directions. Formally, if $S_r$ is the set of points of distance of less than $r$ from some subspace of $\bbR^d$ then we studied potentials satisfying the \emph{subspace Enss condition}
\begin{align*}
    \|\chi_{S_r^c}V\|_{\textrm{op}} \in L^1([0,\infty),dr)
\end{align*}
In this setting, we showed that $\Omega^-$ exists for all $\psi \in \calH$ and that the orthogonal complement of its range is given by the set of surface states
\begin{align*}
    \calH_{\textrm{sur}}=\{\psi \in \calH \mid \forall v>0, \lim\limits_{t\rightarrow \infty}\|\chi_{S_{vt}^c}e^{-itH} \psi\|= 0\}
\end{align*}
so that
\begin{align*}
    \calH=\Ran(\Omega^-)\oplus\calH_\textrm{sur}
\end{align*}
Thus, even though asymptotic completeness in the sense of (\ref{asymptoticCompleteness}) does not generally hold in this setting, we were still able to provide a dynamical characterization of the non-scattering states, albeit not a spectral one. That work proceeds via the phase space scattering methodology of Enss but makes heavy use of the subspace structure of the potential. Naturally then, one can ask whether it is feasible to obtain a similar result for more general geometries, which is the object of this paper.\par
To study scattering, we must choose a potential which will admit asymptotically free trajectories. Classically, one expects that a particle may escape along some ray so long as the strength of the potential attenuates fast enough along that ray. Here, we study the quantum analogue of this phenomenon: we let $V$ decay inside a (possibly infinite) collection of \emph{cones}. For a ray in the interior of a cone, the distance to the boundary of the cone increases along the ray. Thus, the use of cones enforces that the effect of the potential must decrease along any classical free trajectory.\par
Finally, we remark that our results are related to the question of how the geometry of the potential effects the spectrum of $H$. There has been much recent progress in this direction, for instance, the study of geometrically-induced bound states. One representative example is \cite{exner2020spectral} in which a condition is given for the existence of bound states due to singular potentials supported on certain curves in $\bbR^2$. Where our theorem applies to these settings, such states will appear in the interacting subspace $\calH_\textrm{int}$.

\section{Model}
As mentioned above, we consider a self-adjoint operator $H$ on $\calH=L^2(\bbR^d)$ of the form
\begin{align}
    H=H_0+V
\end{align} 
 where $H_0=-\frac{1}{2}\Delta$ and $V$ is a real-valued bounded potential that decays inside a collection of cones.\par
 To be more precise, let us first fix some notation. For any $ x\in \bbR^d, \vec{v}\in \bbS^{d-1}$, and $ \gamma\in (0,\pi)$ let
\begin{align*}
	\mathcal{C}_{x,\gamma,\vec{v}}=\{y\in\bbR \mid \braket{(y-x),\vec{v}}>\cos(\gamma)\|y-x\|\}	
\end{align*}
be the open cone with vertex $ x $ in the $ \vec{v} $ direction with aperture $ 2\gamma $. Since we will often work with cones with vertex at the origin, we let $ \mathcal{C}_{\vec{v},\gamma} $ denote $ \mathcal{C}_{0,\vec{v},\gamma} $. Furthermore, for any cone $\mathcal{C}$, let $A_r(\mathcal{C})$ be the set of points a distance greater than $r>0$ from $\calC^c$:
\begin{align*}
    A_r(\calC)=\{x\in\bbR^d\mid d(x,\calC^c)>r\}
\end{align*}
See Figure \ref{fig:cones}. We will use the shorthand $A_r^c(\calC)=[A_r(\calC)]^c$.\par
For some collection of cones $\{\calC_i\}_{i\in\calI}$, let
\begin{align*}
    &\calA_r=\bigcup\limits_{i\in\calI} A_r(\calC_i)
\end{align*}
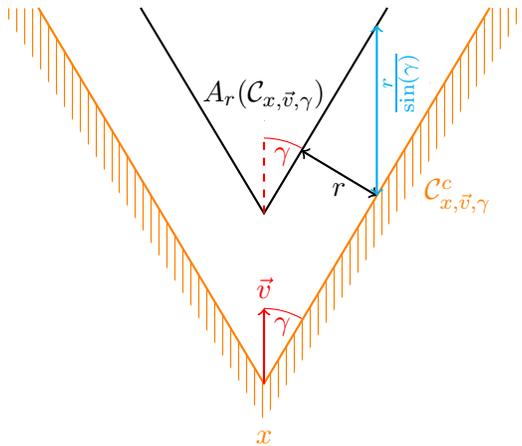
\begin{figure}[ht]
\centering
\begin{tikzpicture}

 \pgfmathsetmacro {\coneendx }{3}
 \pgfmathsetmacro {\coneendy }{5}
 \pgfmathsetmacro {\shift }{0.5}
 \pgfmathsetmacro {\smallconeendx }{41/25}
 \pgfmathsetmacro {\smallconeendy }{34/15}
 \pgfmathsetmacro {\pnty }{31/10}
\draw[white,pattern=vertical lines, pattern color=orange] (\coneendx,\coneendy) -- (0,0) -- (-\coneendx,\coneendy) -- (-\coneendx-\shift ,\coneendy) -- (0,-\shift ) -- (\coneendx+\shift ,\coneendy); 
\draw[orange,thick] (\coneendx,\coneendy) -- (0,0) -- (-\coneendx,\coneendy);     
\filldraw[orange] (\coneendx/2+\shift ,\coneendy/2) circle (0pt) node[anchor=west] {$\calC_{x,\vec{v},\gamma}^c$};
\draw[orange] (0,-\shift ) circle (0pt) node[anchor=north] {$x$};

\draw[red] (0,1) arc (90:60:1);
\filldraw[red] (0,\shift) circle (0pt) node[anchor=south west] {$\gamma$};

\draw[red,->,thick] (0,0) -- (0,1);
\filldraw[red] (0,1) circle (0pt) node[anchor=south] {$\vec{v}$};

\draw[black,thick] (\smallconeendx,5) -- (0,\smallconeendy) -- (-\smallconeendx,5);  
\filldraw[black] (0,3.5) circle (0pt) node[anchor=south] {$A_{r}(\calC_{x,\vec{v},\gamma})$};

\draw[red,dashed,thick] (0,\smallconeendy) -- (0,\smallconeendy+1);  
\draw[red] (0,\smallconeendy+1) arc (90:60:1);
\filldraw[red] (0,\smallconeendy+\shift) circle (0pt) node[anchor=south west] {$\gamma$};

\draw[black,<->,thick] (\coneendx/2,\coneendy/2) -- (\shift,\pnty);
\filldraw[black] (1,56/20) circle (0pt) node[anchor=north] {$r$};
\draw[cyan,<->,thick] (1.5,2.5) -- (1.5,143/30);
\filldraw[cyan] (2.2,4) circle (0pt) node[rotate=90,anchor=south] {$\frac{r}{\sin(\gamma)}$};
\end{tikzpicture}
\caption{Illustration of $\calC_{x,\vec{v},\gamma}$ and $A_{r}(\calC_{x,\vec{v},\gamma})$ for $d=2$: in orange we have the complement of $\calC_{x,\vec{v},\gamma}$, which is where the potential is concentrated. In black we have the set $A_{r}(\calC_{x,\vec{v},\gamma})$, in red we indicate $\vec{v}$ and $\gamma$.
}\label{fig:cones}
\end{figure}
We assume that $V$ satisfies the following \emph{generalized Enss condition} with respect to $\{\calC_i\}_{i\in\calI}$:
\begin{align*}
    \|\chi_{\calA_{r}}V\|_{\textrm{op}}\in L^1([0,\infty),dr)
\end{align*}
which should be compared to (\ref{classicalEnssCondition}). Note that $\calA_r$ depends implicitly on the collection $\{\calC_i\}_{i\in\calI}$ and therefore this condition depends only on the geometry of $V$.\par
We will study the scattering properties of $H$ via the (positive time) wave operator $\Omega^-$, which we simply write as $\Omega$. Our results may be easily reformulated for $\Omega^+$, but we focus our attention on the limit $t\rightarrow\infty$. Before stating a precise theorem (see Section \ref{Def}), let us give a few examples of the geometries we plan to consider.
\begin{example}[Single cone]
    It is already interesting to consider $ V $ which decays in some cone $ \mathcal{C} $, that is, $ \{\mathcal{C}_i\}_{i\in \calI} $ consists of a single cone. For such potentials, we will show that $ \Ran(\Omega) $ consists of states which evolve into $ \mathcal{C} $ with momenta lying in $ \mathcal{C} $ whereas $ \Ran(\Omega)^\perp $ consists of states which may interact with $ V $ for arbitrarily long times. These characterizations are microlocal in the sense that they depend on the position and momentum localization of a state.
\end{example}
\begin{example}[Short-range scattering]
    The condition above is closely related to the classical Enss condition (\ref{classicalEnssCondition}) for short-range potentials. One may study short-range potentials in the present setting by writing
	\begin{align*}
		B_r= \bigcup_{\vec{v}\in \bbS^{d-1}}A_r(\mathcal{C}_{\vec{v},\frac{\pi}{2}})
	\end{align*}
	which may be readily verified. 
\end{example}
\begin{example}[Subspace potentials]
	In \cite{black2022scattering}, we studied potentials which are supported near a subspace of $ \bbR^d $, as explained in Section \ref{Intro}. Using the product structure of this geometry, we proved that $ \Omega $ exists for all $\psi \in \calH$ and gave a purely spatial characterization of $ \Ran(\Omega)^\perp $. We will show that one may recover these results in the present setting since it is easy to see that
	\begin{align*}
		S_r= \bigcup_{\vec{v}\in \bbR^k\times \bbS^{d-k-1}}A_r(\mathcal{C}_{\vec{v},\frac{\pi}{2}})
	\end{align*}
	by similar considerations as in the above example.
\end{example}
\begin{example}[Broken subspace]\label{broken}
A variant of the above example is a ``broken subspace,'' written here in $d=2$ for simplicity: consider $\vec{v}_1,\vec{v}_2\in \bbS^1$ and let $\vec{r}_1$ and $\vec{r}_2$ be the rays $\{t\vec{v}_1\mid t\geq 0\}$ and $\{t\vec{v}_2\mid t\geq0\}$, respectively. Then consider $V$ such that
\begin{align*}
    \supp V\subset T_r:= \{x\in\bbR^2\mid d(x,\vec{r}_1\cup \vec{r}_2)<r\}
\end{align*}
We may accommodate such potentials by observing that 
\begin{align*}
    T_r = (\calC_{r\vec{v}_*,\vec{v}_*,\gamma}\cup \calC_{-r\vec{v}_*,-\vec{v}_*,\pi - \gamma})^c
\end{align*}
where $\vec{v}_*=\frac{\vec{v}_1+\vec{v}_2}{\|\vec{v}_1+\vec{v}_2\|}$ and $\gamma$ is half of the (non-obtuse) angle between $\vec{v}_1$ and $\vec{v}_2$, see Figure \ref{fig:broken}.
\begin{figure}[ht]
\centering
\begin{tikzpicture}

 \pgfmathsetmacro {\deg }{45}
 \pgfmathsetmacro {\size }{7}
 \pgfmathsetmacro {\coneendx }{\size*sin(\deg)}
 \pgfmathsetmacro {\coneendy }{\size*cos(\deg)}
 \pgfmathsetmacro {\lowscale }{1.5}
 \pgfmathsetmacro {\scale }{2}
 \pgfmathsetmacro {\scalev }{3}
 \pgfmathsetmacro {\Twoscale }{1.5*\scale}
 \pgfmathsetmacro {\shift }{1.5}
 \pgfmathsetmacro {\sqrshift }{\shift*sqrt(2)}
 \draw[orange,->,thick] (0,0) -- (\coneendx/\scale,\coneendy/\scale);
\draw[orange] (\coneendx/\scale,\coneendy/\scale) circle (0pt) node[anchor=west] {$v_1$};
 \draw[orange,->,thick] (0,0) -- (-\coneendx/\scalev,\coneendy/\scalev);
\draw[orange] (-\coneendx/\scalev,\coneendy/\scalev) circle (0pt) node[anchor=east] {$v_2$};

 \draw[red,dashed,->,thick] (0,0) -- (0,1);  
\filldraw[red] (0,0.7) circle (0pt) node[anchor=east] {$\vec{v}_*$};
\draw[red] (0,1) arc (90:90-\deg:1);
\filldraw[red] (0,0.5) circle (0pt) node[anchor=south west] {$\gamma$};

 \draw[blue] (0,0+\shift) -- (\coneendx/\scale,\coneendy/\scale+\shift);
 \draw[blue] (0,0+\shift) -- (-\coneendx/\scale,\coneendy/\scale+\shift);
 \draw[blue] (0,0-\shift) -- (\coneendx/\lowscale,\coneendy/\lowscale-\shift);
 \draw[blue] (0,0-\shift) -- (-\coneendx/\lowscale,\coneendy/\lowscale-\shift);
 
 \draw[red,dashed,->,thick] (0,0-\shift) -- (0,-\shift-1);  
\filldraw[red] (0,-\shift-0.7) circle (0pt) node[anchor=east] {$-\vec{v}_*$};
\draw[red] (0,-\shift-1) arc (-90:\deg:1);
\filldraw[red] (0+0.5,-\shift-0.5) circle (0pt) node[anchor=south] {$\pi-\gamma$};
 
 \draw[blue,<->] (\coneendx/\Twoscale,\coneendy/\Twoscale-\shift) -- (\coneendx/\Twoscale-\shift,\coneendy/\Twoscale);
\filldraw[blue] (\coneendx/\Twoscale-\shift/2,\coneendy/\Twoscale-\shift/2+0.1) circle (0pt) node[anchor=south] {$2r$};
\filldraw[blue] (-\coneendx/\Twoscale-\shift/2,\coneendy/\Twoscale-\shift/2+0.1) circle (0pt) node[anchor=south] {$V$};
 
\end{tikzpicture}
\caption{The geometry of the broken subspace: in orange we have the vectors $v_1,v_2$, in red we have the vectors $\vec{v}_*$ and $-\vec{v}_*$, in blue we have the outline of $T_r$- which contains $\supp V$. 
}\label{fig:broken}
\end{figure}
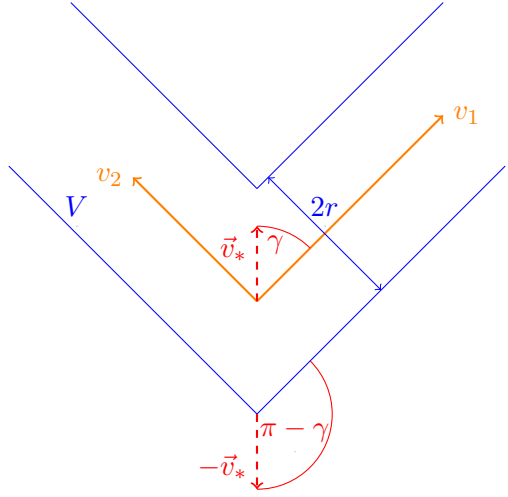

\end{example}
\section{Definitions and Results}\label{Def}
\subsection{Notation and Conventions}
\begin{itemize}
    \item We let $\calH$ denote $L^2(\bbR^d)$ with norm $\|\cdot\|$ and use the convention that its inner product $\braket{\cdot,\cdot}$ is anti-linear in the first argument and linear in the second.
    \item The symbols $\|\cdot\|$ and $\braket{\cdot,\cdot}$ will also be used for the norm and inner product on $\bbR^d$.
    \item $d(\cdot,\cdot)$ is used for the distance between points or subsets of $\bbR^d$.
    \item $B_r$ will mean the ball of radius $r$ centered at the origin in either $\bbR^d$ or $\calH$ depending on context.
    \item For $A\subset \bbR^d$, $A^c$ denotes its complement.
    \item $\chi_A$ will mean the indicator function of $A\subset \bbR^d$.
    \item $A\Subset B$ denotes that $A$ is compactly contained in $B$.
    \item $\calS=\calS(\bbR^d)$, the Schwartz space. 
    \item We use the following convention for the Fourier transform of $f\in\calH$:
    \begin{align*}
        &\hat{f}(\xi)=\calF(f)(\xi)=(2\pi)^{-\frac{d}{2}}\int\limits_{\bbR^d}f(x)e^{-ix\xi} \,dx\\
        &\calF^{-1}(\hat{f})(x)=(2\pi)^{-\frac{d}{2}}\int\limits_{\bbR^d}\hat{f}(\xi)e^{ix\xi} \,d\xi
    \end{align*}
    \item For some cone $\mathcal{C}_{x,\vec{v},\gamma}$ and $r>0$, we define $A_{r}(\mathcal{C}_{x,\vec{v},\gamma})\subset \bbR^d$ to be the set of all points at a distance greater than $r$ from  $\calC_{x,\vec{v},\gamma}^c$
    \begin{align*}
       A_{r}(\calC_{x,\vec{v},\gamma})=\{x\in\bbR^d\mid d(x,\calC_{x,\vec{v},\gamma}^c)>r\}
    \end{align*}
    As explained below,
    \begin{align*}
        A_{r}(\calC_{x,\vec{v},\gamma})=\calC_{x,\vec{v},\gamma}+\frac{r}{\sin(\gamma)}\vec{v}
    \end{align*}
    which we will use to define $A_r$ for $r\leq 0$.
    We will also use the shorthand
    \begin{align*}
        A_{r}^c(\calC_{x,\vec{v},\gamma})=[A_{r}(\calC_{x,\vec{v},\gamma})]^c
    \end{align*}
    \item We will also let 
    \begin{align*}
        &\calA_r=\bigcup\limits_{i\in\calI} A_r(\calC_i)
    \end{align*}
    for $\{\calC_i\}_{i\in\calI}$ some collection of cones. 
    \item For some cone $\mathcal{C}_{x,\vec{v},\gamma}$ and $k>0$ we let
    \begin{align*}
        &\calD_k(\mathcal{C}_{x,\vec{v},\gamma})=\{\psi \in \calS\mid \supp \hat{\psi} \Subset A_k(\calC_{\vec{v},\gamma})\}\\
        &\calD(\mathcal{C}_{x,\vec{v},\gamma})=\{\psi \in \calH\mid \supp \hat{\psi} \subset \calC_{\vec{v},\gamma}\}
    \end{align*}
    \item For the definition of $P_{\delta}(\cdot)$ see Appendix \ref{DaviesProperties}.
    \item $\psi_t$ will always denote the evolution of $\psi$ under $H$ at time $t$: \begin{align*}
        \psi_t = e^{-itH}\psi
    \end{align*}
    \item We will also use the following notation:
    \begin{align*}
        &\Omega(t)=e^{itH}e^{-itH_0}\\
        &\Omega^*(t)=e^{itH_0}e^{-itH}
    \end{align*}
    and we will denote by $\Omega$ the positive time wave operator $\Omega^-$.
    \item $\Ran(\Omega)$ will refer to the range of $\Omega$ on its natural domain $\calD$.
\end{itemize}

\subsection{Definition of the scattering and interacting subspaces}\label{HsurTildeDef}
In order to give the aforementioned microlocal characterizations, we will need a suitable way to describe a state's localization in phase space. To this end, for every $\delta>0$, we define a positive operator-valued measure (POVM), denoted $P_\delta$, on the phase space $\bbR^d_x\times \bbR^d_p$, with the following properties:
\begin{enumerate}
    \item(Observable) $P_\delta(\bbR^{2d})=\id$.
    \item(Momentum localization) Let $B\subset \bbR^d$ and $D\subset \bbR^d$ be Borel sets such that $d(B,D)>\delta$. Then for any $E\subset \bbR^d\times B$ Borel and $\psi \in \calH$ such that $\supp \hat{\psi} \subset D$
    \begin{align*}
        P_\delta(E) \psi =0
    \end{align*}
    \item(Approximate space localization) Let $A\subset \bbR^d$ and $D\subset\bbR^d$ be Borel sets such that $\sloppy{d(D,A)>0}$. Then for any $\ell>0$ there exists some constant $C>0$ so that for all $E\subset A\times \bbR^d$
    \begin{align*}
        \|P_{\delta}(E)\chi_{D}\|_{\textrm{op}}<C[d(A,D)]^{-\ell}
    \end{align*}
    \item(Microlocal non-stationary phase estimate) Let $\mathfrak{C}_t(E)\subset \bbR^{d}$ denote the classically allowed region associated to $E\subset \bbR^{2d}$ at time $t$:
    \begin{align*}
        \mathfrak{C}_t(E)=\{x+tp\mid (x,p)\in E\}
    \end{align*}
    Let $F\subset \bbR^d$ be Borel. For any $\ell>0$ there exists $C>0$ such that 
    \begin{align*}
        \|\chi_{F}e^{-itH_0}P_\delta(E)\|_\mathrm{op}\leq Cd(|t|)^{-\ell}
    \end{align*}
    for all $t$ such that $d(t):=d(\mathfrak{C}_t(E),F)>\delta |t|$.
    \item(Spatial non-stationary phase estimate) Let $\{A_t\}_{t\geq0}$ be a collection of Borel subsets of $\bbR^d$.\par
    Then for any $\varphi \in \calS(\bbR^d)$ Schwartz such that $\supp \hat{\varphi} \Subset D$ Borel, $\ell>0$, and $\varepsilon>0$ there exists some constant $C(\psi,\ell,\varepsilon,\delta)>0$ such that 
    \begin{align*}
        \|P_{\delta}(A_t\times\bbR^d)e^{-itH_0}\varphi \|< Ct^{-\ell}
    \end{align*}
    for all $t$ such that $d(A_{t},tD)>\varepsilon t$.
\end{enumerate}
We refer the reader to \cite{davies1976quantum} for the definition of POVMs and relegate the construction of a POVM satisfying the above properties to Appendix \ref{DaviesProperties}.\par
To specify the domain of $ \Omega $, we let
\begin{align*}
	\calD(\mathcal{C}_{x,\vec{v},\gamma})=\{\psi \in \calH\mid \supp \hat{\psi} \subset \calC_{\vec{v},\gamma}\}
\end{align*}
In particular, $\calD(\mathcal{C}_{x,\vec{v},\gamma})$ is independent of the vertex $x$. For some collection of cones $\{\calC_i\}_{i\in\calI}$, we let
\begin{align*}
    \calD=\overline{\bigcup_{i\in\calI}\calD(\calC_i)}
\end{align*}\par
For $n>0$ and some cone $ \mathcal{C}_{x,\vec{v},\gamma} $, we let the corresponding \emph{outgoing} subset of phase space be the set of points with space coordinates in $A_{n}(\mathcal{C}_{x,\vec{v},\gamma})$ and momentum coordinates in  $\calC_{\vec{v},\gamma}$:
\begin{align*}
	W_{n;\textrm{out}}(\mathcal{C}_{x,\vec{v},\gamma})=\{(x,p)\in\bbR^{2d}\mid x\in A_{n}(\mathcal{C}_{x,\vec{v},\gamma})\text{ and }p\in \mathcal{C}_{\vec{v},\gamma}\}
    \end{align*}
    and let the total outgoing subset be
    \begin{align*}
    	\calW_{n; \textrm{out}}=\bigcup_{i\in \calI}W_{n; \textrm{out}}(\mathcal{C}_i)
    \end{align*}
We also define a variant of $W_{n;\textrm{out}}(\calC)$ which is restricted away from $0$ in the momentum variable:
\begin{align*}
	W_{n,m;\textrm{out}}(\mathcal{C}_{x,\vec{v},\gamma})=\{(x,p)\in\bbR^{2d}\mid x\in A_{n}(\mathcal{C}_{x,\vec{v},\gamma})\text{ and }p\in A_m(\mathcal{C}_{\vec{v},\gamma})\}
\end{align*} 
and it's respective total set
\begin{align*}
    	\calW_{n,m; \textrm{out}}=\bigcup_{i\in \calI}W_{n,m; \textrm{out}}(\mathcal{C}_i)
\end{align*}
 This allows us to define the scattering subspace 
   \begin{align*}
    \calH_{\textrm{scat}}=\{\psi \in \calH\mid \exists v,m,\delta_0>0\text{ so that }\forall\delta\in (0,\delta_0)\lim_{t\rightarrow\infty}\|(P_\delta(\calW_{vt;\mathrm{out}})-\Id)\psi_t\|=0\}
\end{align*}
which we will prove below is dense in $\Ran(\Omega):=\Omega(\calD)$.\par
\begin{remark}
The above characterization of $\Ran(\Omega)$ is similar to those given in \cite{kitadaYajima} and \cite{yoneyama} in the short-range setting.
\end{remark}
We also define the interacting subspace
\begin{align*}
	\calH_{\mathrm{int}}=\{\psi  \in \calH\mid  \forall v,m>0, \exists \delta_0>0\text{ so that }\forall \delta\in(0,\delta_0)\,\lim\limits_{t\rightarrow\infty} \|P_\delta(\calW_{vt,m;\mathrm{out}})\psi_t\|=0\}
\end{align*}
which consists of states that can interact with $V$ for arbitrarily long times. We will show that this subspace is equal to $\Ran(\Omega)^\perp$.\par
With these definitions, we may state our main theorem:
\begin{theorem}\label{thms}
\leavevmode
Let $H=H_0+V$ where $H_0=-\frac{1}{2}\Delta$ and $V$ is a real-valued multiplication operator such that
\begin{itemize}
    \item $V\in L^\infty(\bbR^d)$
    \item There exists a collection of cones $ \{\mathcal{C}_i\}_{i\in \calI} $ for which $V$ satisfies the \emph{generalized Enss condition}
    \begin{align}\label{EnssCond}
    	    \|\chi_{\calA_r}V\|_{\mathrm{op}}\in L^1([0,\infty),dr)
    	\end{align}
    \end{itemize}
Then
\begin{enumerate}[label={(\roman*)},itemindent=1em]
	\item \label{existenceTheorem}(Existence) Let $ \calD = \overline{\bigcup\limits_{i\in \calI}\calD(\mathcal{C}_i)} $. For all $\psi \in \calD$ the limit $\Omega\psi$ exists. Furthermore, ${\sigma(H_0)\subset \sigma_\mathrm{ac}(H)}$.
    \item \label{Complete} (Dynamical description of scattering states and their complement) We have 
    \begin{align*}
        &\Omega(\calD)=\overline{\calH_{\mathrm{scat}}}\\
        &\Omega(\calD)^\perp=\calH_{\mathrm{int}}
    \end{align*}
\end{enumerate}
\end{theorem}
We also show that for cones of large enough aperture, there are spatial characterizations of $\Omega(\calD)$ and $\Omega(\calD)^\perp$:
\begin{theorem}\label{spaceOnlyThm}
    Suppose that $\{\calC_i\}_{i\in \calI}$ consists of cones with aperture greater than or equal to $\pi$. Then with $\calD=\overline{\bigcup_{i\in\calI}\calD(C_i)}$ we have that
    \begin{align*}
    &\Omega(\calD)=\overline{\{\psi \in\calH\mid \exists v>0,\lim\limits_{t\rightarrow \infty }\|\chi_{\calA_{vt}^c}\psi_t\|=0\}}\\
    &\Omega(\calD)^\perp=\{\psi \in\calH\mid \forall v>0,\lim\limits_{t\rightarrow \infty }\|\chi_{\calA_{vt}}\psi_t\|=0\}
\end{align*}
\end{theorem}
\section{Existence of the Wave Operator $\Omega$}\label{Existence}
First, we record a geometric fact:
\begin{proposition}\label{AGeo}
We may write $A_r(\mathcal{C}_{x,\vec{v},\gamma})=\mathcal{C}_{x,\vec{v},\gamma}+\frac{r}{\sin(\gamma)}\vec{v}$.
\end{proposition}
\begin{proof}
By projecting to any plane containing $\vec{v}$, the claim is clear from Figure \ref{fig:cones}.
\end{proof}
For $r\leq 0$ we will use the above as the definition of $A_r(\calC)$.
We now use the following direct application of the Corollary to Theorem XI.14 from \cite{RSVol3}:
\begin{lemma}\label{simpleNonstationary}
	Let $u\in \calS$ and let $\calG$ be an open set such that $\supp \hat{u}\Subset \calG$ . Then for any $\ell\in \bbN $, there is a constant $ C>0$ depending on $ \ell,u,$ and $ \mathcal{G} $ so that
	\begin{align*}
		|e^{-itH_0}u(x)|\leq C(1+\|x\|+|t|)^{-\ell}
	\end{align*}
	for all pairs $ (x,t) $ such that $ \frac{x}{t}\not\in \mathcal{G} $.
\end{lemma}
 We let
 \begin{align*}
   &\calD_k(\mathcal{C}_{x,\vec{v},\gamma})=\{\psi \in \calS\mid \supp \hat{\psi} \Subset A_k(\calC_{\vec{v},\gamma})\}
 \end{align*}
Note that the set $\calD_k(\mathcal{C}_{x,\vec{v},\gamma})$ is independent of the vertex $x$ and that $\bigcup_{k>0}\calD_k(\mathcal{C}_{x,\vec{v},\gamma})$ is dense in $\calD(\mathcal{C}_{x,\vec{v},\gamma})$.\par
We use this to prove the following proposition which will be useful here and in the sequel.
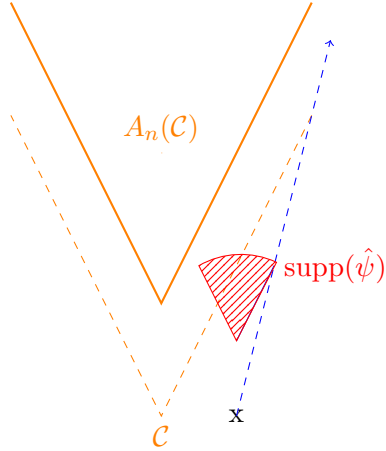
\begin{figure}[ht]
\centering
\begin{tikzpicture}
 \pgfmathsetmacro {\coneendx }{2}
 \pgfmathsetmacro {\coneendy }{4}
 \pgfmathsetmacro {\radius }{sqrt(\coneendy*\coneendy+\coneendx*\coneendx)}
 \pgfmathsetmacro {\scale }{4}
 \pgfmathsetmacro {\pnty }{0}
 \pgfmathsetmacro {\pntx }{1}
 \pgfmathsetmacro {\An }{1.5}
 \pgfmathsetmacro {\Ak }{1}
 \pgfmathsetmacro {\time }{2.5}
\draw[orange,dashed] (\coneendx,\coneendy) -- (0,0) -- (-\coneendx,\coneendy); 
\filldraw[orange] (0,0) circle (0pt) node[anchor=north] {$\mathcal{C}$};

\draw[orange,thick] (\coneendx,\coneendy+\An) -- (0,0+\An) -- (-\coneendx,\coneendy+\An); 

\filldraw[orange] (0,2+\An) circle (0pt) node[anchor=south] {$A_n(\mathcal{C})$};

\draw (\pntx,\pnty) node {x};
\draw[blue,dashed,->] (\pntx,\pnty)-- (\pntx+\time*\coneendx/\scale,\pnty+\time*\coneendy/\scale+\time*\Ak);

\draw[blue] (\pntx,\pnty+\Ak) -- (\pntx+\coneendx/\scale,\pnty+\coneendy/\scale+\Ak);
\draw[red,pattern=north east lines, pattern color=red] (\pntx,\pnty+\Ak) -- (\pntx-\coneendx/\scale,\pnty+\coneendy/\scale+\Ak) arc(120:65:\radius/\scale) -- cycle;
\filldraw[red] (\pntx+\coneendx/\scale,\pnty+\coneendy/\scale+\Ak) circle (0pt) node[anchor=west] {$\supp (\hat{\psi})$};

\end{tikzpicture}
\caption{Illustration of the momentum of $\hat{\psi}$, in red,  with respect to $A_n(\mathcal{C})$, in orange. The dashed blue line corresponds to a classic trajectory from $x$ with momentum at the edge of the red cone.}\label{fig:Existence}
\end{figure}
\begin{proposition}\label{classicnonstationary}
Let $\calC$ be any cone and $k>0$. Then there exists $c(\calC)>0$ such that for all $\psi \in \calD_k(\calC)$ and any $\ell>0$ there exists $C(\psi,\ell)>0$ such that
\begin{align*}
     \|\chi_{A_{n}^c(\mathcal{C})}e^{-itH_0}\psi\|\leq C(1+|t|)^{-\ell}
\end{align*}
for any pair of $(n,t) \in \bbR^2$  satisfying
\begin{align}\label{classicnonstationaryEqn}
    c<kt-n
\end{align}
\end{proposition}
\begin{proof}
	Write $\calC=\calC_{x,\vec{v},\gamma}$. In order to apply Lemma \ref{simpleNonstationary}, we take $\calG=A_k(\calC_{\vec{v},\gamma})$. Thus, we must show that so long as $kt-n$ is sufficiently large, for all $ y\in A_n^c(\mathcal{C}_{x,\vec{v},\gamma}) $,  we have that $ \frac{y}{t} \in A_k^c(\mathcal{C}_{\vec{v},\gamma}) $ or equivalently
    \begin{align*}
    	tA_k(\mathcal{C}_{\vec{v},\gamma})\subset A_n(\mathcal{C}_{x,\vec{v},\gamma})
    \end{align*}
    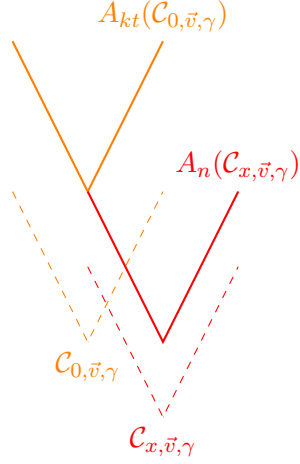
\begin{figure}[ht]
    \centering
    \begin{tikzpicture}[scale=1]
     \pgfmathsetmacro {\pntx }{1}
     \pgfmathsetmacro {\An }{1}
     \pgfmathsetmacro {\Ak }{2}
     \pgfmathsetmacro {\coneendx }{1}
     \pgfmathsetmacro {\coneendy }{2}
    \draw[orange,dashed] (-\coneendx,\coneendy) -- (0,0) -- (\coneendx,\coneendy); 
    \filldraw[orange] (0,0) circle (0pt) node[anchor=north] {$\mathcal{C}_{0,\vec{v},\gamma}$};
    
    \draw[orange,thick] (-\coneendx,\coneendy+\Ak) -- (0,0+\Ak) -- (\coneendx,\coneendy+\Ak); 
    \filldraw[orange] (\coneendx,\coneendy+\Ak) circle (0pt) node[anchor=south] {$A_{kt}(\mathcal{C}_{0,\vec{v},\gamma})$};
    
    \draw[red,thick] (\pntx-\coneendx,\coneendy) -- (\pntx+0,0) -- (\pntx+\coneendx,\coneendy); 
    \filldraw[red]  (\pntx+\coneendx,\coneendy) circle (0pt) node[anchor=south] {$A_n(\mathcal{C}_{x,\vec{v},\gamma})$};
    
    \draw[red,dashed] (\pntx-\coneendx,\coneendy-\An) -- (\pntx+0,0-\An) -- (\pntx+\coneendx,\coneendy-\An); 
    \filldraw[red] (\pntx+0,0-\An) circle (0pt) node[anchor=north] {$\mathcal{C}_{x,\vec{v},\gamma}$};
    \end{tikzpicture}
    \caption{Illustration of the inclusion (\ref{ConeInclsion})}\label{fig:ExInc}
    \end{figure}
    Using Proposition \ref{AGeo} and the fact that $ \mathcal{C}_{\vec{v},\gamma} $ is invariant under scaling, we see that we must show that
    \begin{align}\label{ConeInclsion}
    	\mathcal{C}_{\vec{v},\gamma}+\frac{kt}{\sin(\gamma)}\vec{v}&\subset x+\mathcal{C}_{\vec{v},\gamma}+\frac{n}{\sin(\gamma)}\vec{v}\iff \mathcal{C}_{\vec{v},\gamma}+\frac{kt-n}{\sin(\gamma)}\vec{v}-x\subset \mathcal{C}_{\vec{v},\gamma}
    \end{align}
    In words, we must show that the cone $ \mathcal{C}_{\vec{v},\gamma} $ shifted by the vector $ \frac{kt-n}{\sin(\gamma)}\vec{v}-x $ is contained in $ \mathcal{C}_{\vec{v},\gamma} $, which will be the case as long as this vector lies in the cone. But this is clearly true if $ kt-n $ is large enough with respect to fixed $x$ i.e if (\ref{classicnonstationaryEqn}) holds.\par
    Therefore, we may apply Lemma \ref{simpleNonstationary}, to see that for any $\ell>0$
    \begin{align*}
    	|e^{-itH_0}\psi(y)|\leq C(1+\|y\|+|t|)^{-\ell}
    \end{align*}
    for all $y\in A_n^c(\mathcal{C}_{x,\vec{v},\gamma})$ where $C$ is independent of $y$ and $t$. 
    Choosing $\ell$ large enough, we get that
    \begin{align} \label{ExistEq}
    \begin{split}
    	\|\chi_{A_n^c(\mathcal{C}_{x,\vec{v},\gamma})}e^{-itH_0}\psi\|^2&\leq C  \int\limits_{ A_n^c(\mathcal{C}_{x,\vec{v},\gamma})}(1+\|y\|+|t|)^{-\ell}\,dx<C(1+|t|)^{-\ell+d}
    	\end{split} 
    \end{align}
    as needed.
\end{proof}

This is already enough to prove the existence of the wave operators:
\begin{proof}[Proof of part \ref{existenceTheorem} of Theorem \ref{thms}]
 By Cook's method (see \cite{RSVol3} 
 Theorem XI.4), it suffices to show that for all $\psi$ in some dense subset of $\calD=\overline{\bigcup\limits_{i\in \calI}\calD(\mathcal{C}_{i})}$ 
\begin{align*}
	\int\limits_0^\infty \|Ve^{-itH_0}\psi\|\,dt<\infty
\end{align*}
We will take as our dense subset $\bigcup\limits_{i\in\calI}\bigcup\limits_{k>0}\calD_k(\mathcal{C}_{i})$.\par
For any $i\in \calI$ and any $k>0$, write $\calC_i=\calC_{x,\vec{v},\gamma}$, let $0<\varepsilon<k$, and let $\psi\in \calD_k(\mathcal{C}_{x,\vec{v},\gamma})$. We can then write
\begin{align*}
    \|Ve^{-itH_0}\psi\|&\leq \|V\chi_{\calA_{\varepsilon t}}e^{-itH_0}\psi\|+\|V\chi_{\calA_{\varepsilon t}^c}e^{-itH_0}\psi\|\\
    &\leq \|V\chi_{\calA_{\varepsilon t}}\|_{\textrm{op}}\|\psi\|+M\|\chi_{A_{\varepsilon t}^c (\calC_i)}e^{-itH_0}\psi\|
\end{align*}
as $\calA^c_{\varepsilon t}\subset A_{\varepsilon t}^c (\calC_i)$. The first term is $L^1([0,\infty),dt)$ by the assumption (\ref{EnssCond}) whereas we will estimate the second term via Proposition \ref{classicnonstationary}. For this, let $c_0=c_0(\mathcal{C}_{x,\vec{v},\gamma})$ be the constant from Proposition \ref{classicnonstationary} and let $T_0=\frac{c_0}{k-\varepsilon}$. By Proposition \ref{classicnonstationary} with $n=\varepsilon t$, we see that for any $\ell>0$ and $t>T_0$ there is some $C>0$ so that
\begin{align} \label{ExistEq2}
\begin{split}
	\|\chi_{A_{\varepsilon t}^c(\mathcal{C}_{i})}e^{-itH_0}\psi\|\leq C (1+t)^{-\ell}
	\end{split} 
\end{align}
uniformly in $t$.
It follows immediately that $\|Ve^{-itH_0}\psi\|$ is integrable on $[0,\infty)$ as needed.\par
Furthermore, since $\sigma(H_0\vert_{\calD})=\sigma(H_0)$, the intertwining property of $\Omega$ implies that
${\sigma(H_0)\subset \sigma_\textrm{ac}(H)}$ as claimed.
\end{proof}
\section{Description of $\Ran(\Omega)$ and $\Ran(\Omega)^\perp$}
In this section we give descriptions of $\Ran(\Omega)$ and its orthogonal complement in terms of the dynamics of $H$. In particular, we show that states in these subspaces may be characterized by their location in phase space as $t\rightarrow\infty$. Here, as before, $\Ran(\Omega)$ indicates the range of $\Omega$ on its natural domain $\calD$.\par

\subsection{Characterizing $\Ran(\Omega)$}
\begin{figure}[ht]
\centering
\begin{tikzpicture}
 \pgfmathsetmacro {\coneendx }{3}
 \pgfmathsetmacro {\coneendy }{4}
 \pgfmathsetmacro {\smallconeendx }{sqrt(3/7)}
 \pgfmathsetmacro {\smallconeendy }{4/3*\smallconeendx }
 \pgfmathsetmacro {\pnty }{4}
 \pgfmathsetmacro {\pntx }{0.2}
 \pgfmathsetmacro {\An }{2}
 \pgfmathsetmacro {\Ak }{0.2}
 \pgfmathsetmacro {\deg }{37}
\draw[orange,dashed] (\coneendx,\coneendy) -- (0,0) -- (-\coneendx,\coneendy);     
\filldraw[orange]  (0,0) circle (0pt) node[anchor=north] {$\calC_i$};

\draw[red,dashed,thick] (0,0) -- (0,1);
\draw[red] (0,1) arc (90:60:1);
\filldraw[red] (-0.1,0.7) circle (0pt) node[anchor=west] {$\gamma$};

\draw[black,thick] (\coneendx,\coneendy+\An) -- (0,0+\An) -- (-\coneendx,\coneendy+\An);      
\filldraw[black]  (0,0+\An) circle (0pt) node[anchor=north] {$A_n(\calC_i)$};

\draw (\pntx,\pnty) node {x};
\draw[red,pattern=north east lines, pattern color=red] (\pntx,\pnty) -- (\pntx-\smallconeendx,\pnty+\smallconeendy) arc(90+\deg:90-\deg:1) -- cycle;
\filldraw[red] (\pntx+\smallconeendx,\pnty+\smallconeendy+\Ak) circle (0pt) node[anchor=west] {$W_{n;\textrm{out}}$};

\draw[blue,pattern=north east lines, pattern color=blue] (\pntx,\pnty+\Ak) -- (\pntx-\smallconeendx,\pnty+\smallconeendy+\Ak) arc(90+\deg:90-\deg:1) -- cycle;
\filldraw[blue] (\pntx-\smallconeendx,\pnty+\smallconeendy+\Ak) circle (0pt) node[anchor=east] {$W_{n,m;\textrm{out}}$};

\end{tikzpicture}
\caption{Illustration of the phase space sets $W_{n;\textrm{out}}(\calC_i)$ and $W_{n,m;\textrm{out}}(\calC_i)$: space coordinates are inside the black cone while momentum coordinates point inside the red/blue cone, respectively. }\label{fig:momentumsets}
\end{figure}
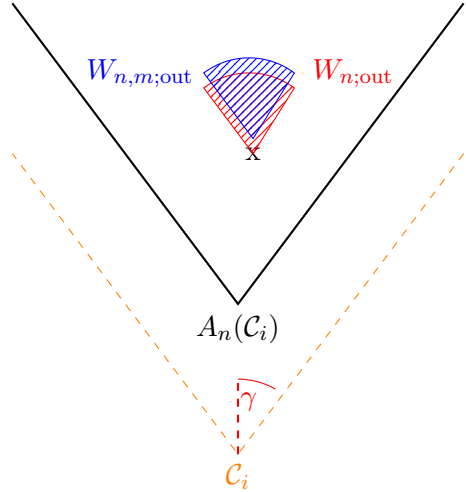

We again record some purely geometric facts:
\begin{proposition}\label{coneOnlyGeo}
\leavevmode
\begin{enumerate}
	\item\label{coneOnlyGeoClaim1} Let $ \mathcal{C} $ be any cone. For any $n,t,r\geq 0$ and $m\in\bbR$		\begin{align*}
		d(\mathfrak{C}_t(W_{n,m;\mathrm{out}}(\mathcal{C})),A_r^c(\calC))\geq n+mt-r
		\end{align*}
	\item\label{coneOnlyGeoClaim2} For any $n,t,r\geq 0$ and $m\in\bbR $
		\begin{align*}
		d(\mathfrak{C}_t(W_{n,m; \mathrm{out}}(\mathcal{C})),\calA_{r}^c)\geq n+mt-r
		\end{align*}
\end{enumerate}
\end{proposition}
\begin{proof}
	The proof of (\ref{coneOnlyGeoClaim1}) follows from the fact that for any $t\geq 0$
\begin{align*}
	A_n(\mathcal{C}_{x,\vec{v},\gamma})+tA_m(\mathcal{C}_{\vec{v},\gamma})= A_{n+tm}(\mathcal{C}_{x,\vec{v},\gamma})
\end{align*}
Indeed, because $t\calC_{\vec{v},\gamma}=\calC_{\vec{v},\gamma}$ and $A_n(\mathcal{C}_{x,\vec{v},\gamma})=x+\calC_{\vec{v},\gamma}+\frac{n}{\sin(\gamma)}\vec{v}$ we see that
\begin{align}
	\begin{split}
	&A_n(\mathcal{C}_{x,\vec{v},\gamma})+tA_m(\mathcal{C}_{\vec{v},\gamma})=(x+\calC_{\vec{v},\gamma}+\frac{n}{\sin(\gamma)}\vec{v})+t(\mathcal{C}_{\vec{v},\gamma}+\frac{m}{\sin(\gamma)}\vec{v})\\
	&=x+\mathcal{C}_{\vec{v},\gamma}+(\frac{n}{\sin(\gamma)}+\frac{tm}{\sin(\gamma)})\vec{v}=A_{n+tm}(\mathcal{C}_{x,\vec{v},\gamma})
	\end{split}
\end{align}
where we have used that $\calC_\gamma$ is additive. The claim now follows from the definition of $A_r(\calC)$.\par
The proof of (\ref{coneOnlyGeoClaim2}) is now immediate because for all $ i \in \calI $, $ \calA_{r}^c\subset A_r^c(\mathcal{C}_i) $.
\end{proof}
With this in hand, we can prove the main technical estimate in the proof of Theorem \ref{thms} \ref{Complete}.
\begin{lemma}\label{MainEstNoMB}
For any $v>0$, $\delta<v$, $0<\varepsilon<\frac{v-\delta}{2}$ and $\ell>0$ there exits $C>0$ so that for any $2t\geq s\geq t>0$
\begin{align}\label{OmegaEst}
      \|(\Omega(s-t)-\Id)P_\delta(\calW_{vt;\mathrm{out}})\|_{\mathrm{op}}\leq  \int\limits_t^{s}\|V\chi_{\calA_{\varepsilon w}}\|_{\mathrm{op}}\, dw+Ct^{-\ell}
\end{align}
For any $v,m>0$ and $\ell>0$ if $\delta<m$, and $0<\varepsilon<\min(m-\delta,v)$ there exits $C>0$ so that for any $s\geq t>0$
\begin{align}\label{OmegaEstWithm}
      \|(\Omega(s-t)-\Id)P_\delta(\calW_{vt,m;\mathrm{out}})\|_{\mathrm{op}}\leq  \int\limits_t^{\infty}\|V\chi_{\calA_{\varepsilon w}}\|_{\mathrm{op}}\, dw+ Ct^{-\ell}
\end{align}
\end{lemma}
\begin{proof}
We use the identity
\begin{align*}
    \Omega(s-t)-\Id = \int\limits_0^{s-t}e^{iwH}(H-H_0)e^{-iwH_0}\,dw
\end{align*}
by writing, for $0<\varepsilon<\frac{v-\delta }{2}$
\begin{align*}
     &\|\int\limits_0^{s-t} e^{iwH}(H-H_0)e^{-iwH_0} P_\delta(\calW_{vt;\textrm{out}})\,dw\|_{\textrm{op}}\leq  \int\limits_0^{s-t}\|Ve^{-iwH_0}P_\delta(\calW_{vt;\textrm{out}})\|_{\textrm{op}}\, dw\\
    &\leq \int\limits_0^{s-t}\|V\chi_{\calA_{\varepsilon (w+t)}}\|_{\textrm{op}}\, dw+ M\int\limits_0^{s-t}\|\chi_{\calA_{\varepsilon (w+t)}^c}e^{-iwH_0}P_\delta(\calW_{vt;\textrm{out}})\|_{\textrm{op}}\, dw
\end{align*}
From the microlocal non-stationary phase estimate on $P_\delta$ (Lemma \ref{generalnonstationary}) and Proposition \ref{coneOnlyGeo} with $m=0$, we see that
\begin{align*}
    \|\chi_{\calA_{\varepsilon (w+t)}^c}e^{-iwH_0}P_\delta(\calW_{vt;\textrm{out}})\|_{\textrm{op}}< C[(v-2\varepsilon)t]^{-(\ell+1)}
\end{align*}
for all $w<t$ since by Proposition \ref{coneOnlyGeo} we have
\begin{align*}
    d(\mathfrak{C}_w(\calW_{vt;\textrm{out}}),\calA_{\varepsilon (w+t)}^c)\geq vt-(w+t)\varepsilon>(v-2\varepsilon )w
\end{align*}
which is in turn greater than $\delta w$ because $\varepsilon<\frac{v-\delta}{2}$.
Therefore, because $s-t\leq t$
\begin{align}\label{intEstofWaveOp}
    \int\limits_0^{s-t}\|\chi_{\calA_{\varepsilon (w+t)}^c}e^{-iwH_0}P_\delta(\calW_{vt;\textrm{out}})\|_{\textrm{op}}\, dw\leq C\frac{s-t}{((v-2\varepsilon)t)^{l}}\leq Ct^{-\ell+1}
\end{align}
In summary, we see that for any $t>0$ and any $s\in[t, 2t]$
\begin{align*}
    \|(\Omega(s-t)-\Id)P_\delta(\calW_{vt;\textrm{out}})\|_{\textrm{op}}\leq  \int\limits_0^{s-t}\|V\chi_{\calA_{\varepsilon(w+t)}}\|_{\textrm{op}}\, dw+ Ct^{-\ell}= \int\limits_t^{s}\|V\chi_{\calA_{\varepsilon w}}\|_{\textrm{op}}\, dw+ Ct^{-\ell}
\end{align*}
for some constant $C$ independent of $t$ and $s$.\par
If we replace $\calW_{vt;\textrm{out}}$ with $\calW_{vt,m;\textrm{out}}$, again from Lemma \ref{generalnonstationary} and Proposition \ref{coneOnlyGeo} with $m>0$, we see that, for $\varepsilon<\min(m-\delta,v)$
\begin{align*}
	\|\chi_{\calA_{\varepsilon (w+t)}^c}e^{iwH_0}P_\delta(\calW_{vt,m;\textrm{out}})\|_{\textrm{op}}< C((v-\varepsilon)t+(m-\varepsilon)w)^{-(\ell+1)}
\end{align*}
for all $w>0$ since by Proposition \ref{coneOnlyGeo}, for $\delta<m$ and $\varepsilon<\min(m-\delta,v)$
\begin{align*}
    d(\mathfrak{C}_w(\calW_{vt,m;\textrm{out}}),\calA_{\varepsilon (w+t)}^c)=(v-\varepsilon)t+(m-\varepsilon)w>(m-\varepsilon)w>\delta w
\end{align*} 
Therefore,
\begin{align*}
	\int\limits_0^{s-t}\|\chi_{\calA_{\varepsilon (w+t)}^c}e^{iwH_0}P_\delta(\calW_{vt,m;\textrm{out}})\|_{\textrm{op}}\, dw\leq C\int\limits_0^{s-t}((v-\varepsilon)t+(m-\varepsilon)w)^{-\ell-1}\, dw\leq C((v-\varepsilon)t)^{-\ell}
\end{align*}
In summary, we see that for any $s>t>0$ 
\begin{align*}
    \|(\Omega(s-t)-\Id)P_\delta(\calW_{vt,m;\textrm{out}})\|_{\textrm{op}}\leq \int\limits_t^{\infty}\|V\chi_{\calA_{\varepsilon w}}\|_{\textrm{op}}\, dw+ Ct^{-\ell}
\end{align*}
for some constant $C$ independent of $t$ and $s$.\par

\end{proof}
Recall that
\begin{align*}
    \calH_{\textrm{scat}}=\{\psi \in \calH\mid \exists v,m,\delta_0>0\text{ so that }\forall\delta\in (0,\delta_0)\lim_{t\rightarrow\infty}\|(P_\delta(\calW_{vt,m;\mathrm{out}})-\Id)\psi_t\|=0\}
\end{align*}
\begin{theorem}\label{OmegaStarArg}
    Let $\psi\in \calH_{\mathrm{scat}}$. Then $\Omega^*\psi$ exists and is in $\calD$, or equivalently $\psi\in\Omega(\calD)$.\par
    Alternatively, if for some $ v>0,\delta<v$ and $\ell>1$ there exists $C>0$ so that
    \begin{align*}
    \|(P_{\delta}(\calW_{vt;\mathrm{out}})-\Id)\psi_{t}\|<Ct^{-\ell}
    \end{align*}
    for all $t>0$ then $ \Omega^*\psi $ exists and lies in $\calD$.
\end{theorem}
\begin{proof}
We show that
\begin{align*}
    \Omega^*(t):=e^{itH_0}e^{-itH}
\end{align*}
is Cauchy as $t\rightarrow\infty$.
For that, fix $t>0$ and suppose that $t\leq s $. Observe that
\begin{align*}
    \|(\Omega^*(s)-\Omega^*(t))\psi\|=\|(\Omega(s-t)-\Id)\psi_t\|
\end{align*}
which comes from multiplying by $e^{-itH}\Omega(s)$ and the identity
\begin{align*}
    e^{-itH}\Omega(s)\Omega^{*}(t)=\Omega(s-t)e^{-itH}
\end{align*}
Since $\psi \in \calH_{\textrm{scat}}$, there is some $v,m,\delta_0>0$ such that for any $\delta<\delta_0$
\begin{align*}
    \|(P_\delta(\calW_{vt,m;\textrm{out}})-\Id)\psi_t\| =o(1)
\end{align*}
as $t\rightarrow\infty$. For these $v,m>0$ choose $\delta<\min(m,\delta_0)$, so  we may write
\begin{align*}
   \|(\Omega(s-t)-\Id)\psi_t\| &\leq \|(\Omega(s-t)-\Id)P_\delta(\calW_{vt,m;\textrm{out}})\psi_t\|+ \|(\Omega(s-t)-\Id)(P_\delta(\calW_{vt,m;\textrm{out}})-\Id)\psi_t\|\\
   &= \|(\Omega(s-t)-\Id)P_\delta(\calW_{vt,m;\textrm{out}})\psi_t\|+o(1)
\end{align*}
as $t\rightarrow\infty$. By using Lemma \ref{MainEstNoMB}, we conclude that, for $0<\varepsilon<\min(m-\delta,v)$
\begin{align*}
    \|(\Omega^*(s)-\Omega^*(t))\psi\|\leq Ct^{-\ell}+\int\limits_t^{\infty}\|V\chi_{\calA_{\varepsilon w}}\|_{\textrm{op}}\, dw\|\psi\|+o(1)
\end{align*}
for some constant $C$ independent of $t$ and $s$. The second term decays with $t$, by assumption (\ref{EnssCond}), and thus the entire expression goes to $0$ as $t\rightarrow\infty$.\par
This shows that $\Omega^*\psi$ exists, so to see that it lies in $\calD$ first note that for $\delta<m$, $P_\delta(\calW_{vt,m;\textrm{out}})\psi\in\calD$ by Proposition \ref{psupport} and thus so does $e^{itH_0}P_\delta(\calW_{vt,m;\textrm{out}})$. Now observe
\begin{align*}
    \|\Omega^*(t)\psi-e^{itH_0}P_\delta(\calW_{vt,m;\textrm{out}})\psi_t\|=\|\psi_t-P_\delta(\calW_{vt,m;\textrm{out}})\psi_t\|
\end{align*}
which goes to $0$ be assumption so that the claim follows because $\calD$ is closed.\par
To see the second claim, for $t>0$ and $s\in [t,2t]$, write as before, for the given $v$ and $\delta<v$
\begin{align*}
   \|(\Omega(s-t)-\Id)\psi_t\| &\leq \|(\Omega(s-t)-\Id)P_\delta(\calW_{vt;\textrm{out}})\psi_t\|+ \|(\Omega(s-t)-\Id)(P_\delta(\calW_{vt;\textrm{out}})-\Id)\psi_t\|\\
   &\leq \|(\Omega(s-t)-\Id)P_\delta(\calW_{vt;\textrm{out}})\psi_t\|+Ct^{-\ell}
\end{align*}
and again apply Lemma \ref{MainEstNoMB} to see that
\begin{align*}
    \|(\Omega^*(s)-\Omega^*(t))\psi\|\leq Ct^{-\ell}+ \int\limits_t^{s}\|V\chi_{\calA_{\varepsilon w}}\|_{\textrm{op}}\, dw\|\psi\|
\end{align*}
for some constant $C$ independent of $t$ and $s$.\par
To conclude, for any $s\geq t>0$, fix $N$ so that $s\in [2^Nt,2^{N+1}t]$ and then write
\begin{align*}
    \|(\Omega^*(s)-\Omega^*(t))\psi\|&\leq \sum_{n=0}^{N-1}\|(\Omega^*(2^{n+1}t)-\Omega^*(2^{n}t))\psi\|+\|(\Omega^*(s)-\Omega^*(2^Nt))\psi\|\\
    &\leq C\sum_{n=0}^N (2^nt)^{-\ell}+\sum_{n=0}^{N-1} \int\limits_{2^nt}^{2^{n+1}t}\|V\chi_{\calA_{\varepsilon w}}\|_{\textrm{op}}\, dw\|\psi\|+\int\limits_{2^Nt}^{s}\|V\chi_{\calA_{\varepsilon w}}\|_{\textrm{op}}\, dw\|\psi\|\\
    &\leq Ct^{-\ell}+\int\limits_{t}^{s}\|V\chi_{\calA_{\varepsilon w}}\|_{\textrm{op}}\,dw \|\psi\|
\end{align*}
where $C$ does not depend on $N$. This, combined with condition (\ref{EnssCond}), prove that $\Omega^*\psi$ exists.\par
To see that $\Omega^*\psi$ lies in $\calD$, we proceed as before by noting that for any $\delta>0$, Proposition \ref{psupport} shows that
\begin{align*}
    \supp \calF(P_{\delta}(\calW_{vt;\textrm{out}})\psi)\subset B_\delta+\bigcup_{(x,\vec{v},\gamma)\in\calI}\calC_{\vec{v},\gamma}
\end{align*}
Now we can write
\begin{align*}
    \|\Omega^*\psi-e^{itH_0}P_\delta(\calW_{vt;\textrm{out}})\psi_t\|\leq \|\Omega^*\psi-\Omega^*(t)\psi\|+\|\Omega^*(t)\psi-e^{itH_0}P_\delta(\calW_{vt;\textrm{out}})\psi_t\|
\end{align*}
By taking the limit $t\rightarrow\infty $ we see that for any $v>\delta>0$
\begin{align*}
    \supp \widehat{\Omega^*\psi}\subset \overline{B_\delta+\bigcup_{(x,\vec{v},\gamma)\in\calI}\calC_{\vec{v},\gamma}}
\end{align*}
Varying over all $\delta>0$, we conclude that $\Omega^*\psi\in \calD$, as desired.
\end{proof}
Having shown that $\calH_{\textrm{scat}}\subset \Ran(\Omega)$, we now show that $\calH_{\textrm{scat}}$ is dense in this subspace. For this we will start with a lemma:
\begin{lemma}\label{RanDkLemma}
 Let $\psi \in \Omega(\calD_k(\mathcal{C}_i))$ for some $i\in\calI$ and $k>0$. Then there is some $T_0=T_0(k,\calC_i)$ such that for any $v,m,\varepsilon,$ and $\delta$ satisfying
\begin{align*}
    &v,\varepsilon\in(0,k)&& 0\leq m<k &&&0<\delta<\min(k-m,\frac{k-v}{2})
\end{align*}
there exists constants $\ell,C>0$ such that
\begin{align}\label{DmClaim}
    \|(P_\delta(\calW_{vt,m;\mathrm{out}})-\Id)\psi_t\|\leq Ct^{-\ell}+\int\limits_t^\infty\|\chi_{A_{\varepsilon s}}V\|_{\mathrm{op}}\|\psi\|\,ds
\end{align}
for all $t>T_0$.
\end{lemma}

\begin{proof}
Let $\psi=\Omega\varphi$ for $\varphi\in \calD_k(\mathcal{C}_{i})$, some $i\in \calI$, and some fixed $k>0$.
It suffices to show that for some choice of parameters as above that for all $t>T_0$
\begin{align}\label{term1}
    \|(P_\delta(\calW_{vt,m;\textrm{out}})-\Id)e^{-itH_0}\varphi\|<Ct^{-\ell}
\end{align}
and
\begin{align}\label{term2}
   \|(\Omega-\Id)e^{-itH_0}\varphi\|\leq Ct^{-\ell}+\int\limits_t^\infty\|\chi_{A_{\varepsilon s}}V\|_{\textrm{op}}\|\psi\|\,ds
\end{align}
in light of the inequality
\begin{align*}
    &\|(P_\delta(\calW_{vt,m;\textrm{out}})-\Id)e^{-itH}\psi\|=\|(P_\delta(\calW_{vt,m;\textrm{out}})-\Id)\Omega e^{-itH_0}\varphi\|\\
    &\leq \|(P_\delta(\calW_{vt,m;\textrm{out}})-\Id)e^{-itH_0}\varphi\|+\|(P_\delta(\calW_{vt,m;\textrm{out}})-\Id)(\Omega-\Id)e^{-itH_0}\varphi\|
\end{align*}
and the fact that $\|P_\delta(W_{vt,m;\textrm{out}})-\Id \|_{\textrm{op}}$ is bounded independently of $t$.\par
The inequality (\ref{term2}) is proven by first choosing $\varepsilon<k$ and writing
\begin{align*}
    \|(\Omega-\Id)e^{-itH_0}\varphi\|&\leq \int\limits_0^\infty\|Ve^{-i(s+t)H_0}\varphi\|\,ds= \int\limits_t^\infty\|Ve^{-isH_0}\varphi\|\,ds\\
    &\leq \int\limits_t^\infty\|\chi_{\calA_{\varepsilon s}}V\|_{\textrm{op}}\|\psi\|\,ds+M\int\limits_t^\infty\|\chi_{\calA_{\varepsilon s}^c}e^{-isH_0}\varphi\|\,ds
\end{align*}
where have used that $\|\varphi\|=\|\psi\|$. Now let $c$ be the constant from Proposition \ref{classicnonstationary} and note that
\begin{align*}
    c<ks-\varepsilon s
\end{align*}
so long as $s>T_1:=\frac{c}{k-\varepsilon}$. Therefore, Proposition \ref{classicnonstationary} with $n=\varepsilon s$ implies the desired inequality for any $t>T_1$.
Therefore, it remains to show (\ref{term1})
for some choice of parameters as above.\par
To see inequality (\ref{term1}), we write $\calC_i=\calC_{x,\vec{v},\gamma}$ and observe that
\begin{align*}
    &\|(P_\delta(\calW_{vt,m;\textrm{out}})-\Id)e^{-itH_0}\varphi\|^2\leq \braket{e^{-itH_0}\varphi, P_\delta(\calW_{vt,m;\textrm{out}}^c)e^{-itH_0}\varphi}\\
    &\leq \braket{e^{-itH_0}\varphi, P_\delta(W_{vt,m;\textrm{out}}^c(\calC_i))e^{-itH_0}\varphi} \leq \|\varphi\| \|P_\delta(W_{vt,m;\textrm{out}}^c(\calC_i))e^{-itH_0}\varphi\|
\end{align*}
Noting that 
\begin{align*}
    &W_{vt,m;\textrm{out}}^c(\calC_i)=A_{vt}^c(\calC_i)\times \bbR^d\bigsqcup A_{vt}(\calC_i)\times A_m^c(\calC_{\vec{v},\gamma})
\end{align*}
and recalling that $\supp\hat{\varphi} \Subset A_k(\calC_{\vec{v},\gamma})$, by the momentum localization properties of $P_\delta$ (Proposition \ref{psupport}) we see that
\begin{align*}
    P_\delta(W_{vt,m;\textrm{out}}^c(\calC_i))\varphi=P_\delta(A_{vt}^c(\calC_i)\times \bbR^d)\varphi
\end{align*}
as $\delta<k-m$ and $d(A_k(\calC_{\vec{v},\gamma}),A_m^c(\calC_i))>k-m$.\par
Next, choose $T_2=\frac{2\|x\|}{k-v}$ so that for any $t>T_2$ 
\begin{align*}
    &d(A_{vt}^c(\calC_i),tA_k(\calC_{\vec{v},\gamma}))>(k-v)t-\|x\|>\frac{k-v}{2}t>\delta t
\end{align*}
Proposition \ref{BasicNonstationaryPd} then implies that for any $\ell>0$ there is some $C>0$ such that
\begin{align*}
    \|P_\delta(W_{vt,m;\textrm{out}}^c(\calC_i))e^{-itH_0}\varphi\|<Ct^{-\ell}
\end{align*}
for all $t>T_2$ which proves (\ref{term2}). We then conclude that the lemma holds with $T_0=\max(T_1,T_2)$.
\end{proof}
\begin{theorem}\label{scatTheorem}
    Suppose that $\psi\in \Ran(\Omega)$. Then $\psi\in \overline{\calH_{\mathrm{scat}}}$.
\end{theorem}
\begin{proof}
Since $\bigcup\limits_{i\in\calI}\bigcup\limits_{k>0}\Omega(\calD_k(\mathcal{C}_{i}))$ is dense in $\Ran(\Omega)$, it suffices to show that $\Omega(\calD_k(\mathcal{C}_{i}))\subset\calH_{\textrm{scat}}$ for all $k>0$.
But this is immediate from Lemma \ref{RanDkLemma} so long as $m,\varepsilon$ and $v$ are chosen appropriately with respect to $k$ and $\delta_0$ is chosen to be less than $\min(k-m,\frac{k-v}{2})$. 
\end{proof}
\begin{remark}
The second claim in Theorem \ref{OmegaStarArg} and the above proof of Theorem \ref{scatTheorem} also show that $\Ran(\Omega)$ may be described without the parameter $m$ as
\begin{align*}
   \Omega(\calD)=\overline{\{\psi \in \calH\mid  \exists v,C,\ell>0,\delta_0>0\text{ so that }\forall \delta\in(0,\delta_0)\text{ and } t>0\,\|(P_\delta(\calW_{vt;\mathrm{out}})-\Id)\psi_t\|<Ct^{-\ell}\}}
\end{align*}
but we prefer the given characterization as  $\calH_{\textrm{scat}}$ because $\calH_{\textrm{int}}$ must be defined in terms of $m$.
\end{remark}
\subsection{Characterizing $\Ran(\Omega)^\perp$}
Recall that
\begin{align*}
      \calH_{\mathrm{int}}=\{\psi  \in \calH\mid  \forall v,m>0, \exists \delta_0>0\text{ so that }\forall \delta\in(0,\delta_0)\,\lim\limits_{t\rightarrow\infty} \|P_\delta(\calW_{vt,m;\mathrm{out}})\psi_t\|=0\}
\end{align*}
\begin{theorem}\label{Orthogonal}
Under the above definition, $\Ran(\Omega)^\perp=\calH_{\mathrm{int}}$.
\end{theorem}
\begin{proof}
For the inclusion $\Ran(\Omega)^\perp\subset \calH_\textrm{int}$, take $\psi \in \Ran(\Omega)^\perp$ and fix any $v,m>0$ and $\delta<m$. For any $s>t>0$ we may write
\begin{align*}
    &\|P_\delta(\calW_{vt,m;\textrm{out}})\psi_t\|^2\leq \braket{P_\delta(\calW_{vt,m;\textrm{out}})\psi_t,\psi_t}\\
    &\leq \braket{\Omega(s-t)P_\delta(\calW_{vt,m;\textrm{out}})\psi_t,\psi_t}+\|\psi\|\|(\Omega(s-t)-\Id)P_\delta(\calW_{vt,m;\textrm{out}})\psi_t\|\\
    &=\braket{e^{itH}\Omega(s-t)P_\delta(\calW_{vt,m;\textrm{out}})\psi_t,\psi}+\|\psi\|\|(\Omega(s-t)-\Id)P_\delta(\calW_{vt,m;\textrm{out}})\psi_t\|\\
    &=\braket{\Omega(s)e^{itH_0}P_\delta(\calW_{vt,m;\textrm{out}})\psi_t,\psi}+\|\psi\|\|(\Omega(s-t)-\Id)P_\delta(\calW_{vt,m;\textrm{out}})\psi_t\|
\end{align*}
where we have used that
\begin{align*}
    e^{itH}\Omega(s-t)=\Omega(s)e^{itH_0}
\end{align*}
Now by applying (\ref{OmegaEstWithm}) from Lemma \ref{MainEstNoMB} to the second term, we get that for any $s\geq t>0$, and $0<\varepsilon<\min(m-\delta,v)$
\begin{align*}
    &\|P_\delta(\calW_{vt,m;\textrm{out}})\psi_t\|^2\leq \braket{\Omega(s)e^{itH_0}P_\delta(\calW_{vt;\textrm{out}})\psi_t,\psi}+Ct^{-\ell}+\int\limits_t^{\infty}\|V\chi_{\calA_{\varepsilon w}}\|_{\textrm{op}}\, dw
\end{align*}
for some constant $C$ that does not depend on $t$ or $s$. Observe that because $\delta<m$, by Proposition \ref{psupport},  $P_\delta(\calW_{vt,m;\textrm{out}})\psi_t$ lies in $\calD$ as does $e^{itH_0}P_\delta(\calW_{vt,m;\textrm{out}})\psi_t$ since the free propagator does not alter the momentum support of a state. Thus, with $t$ fixed, we may take the limit $s\rightarrow \infty $ in the above to obtain
\begin{align*}
    \|P_\delta(\calW_{vt,m;\textrm{out}})\psi_t\|^2&\leq \braket{\Omega e^{itH_0}P_\delta(\calW_{vt,m;\textrm{out}})\psi_t,\psi}+Ct^{-\ell}+\int\limits_t^{\infty}\|V\chi_{\calA_{\varepsilon w}}\|_{\textrm{op}}\, dw\\
    &=Ct^{-\ell}+\int\limits_t^{\infty}\|V\chi_{\calA_{\varepsilon w}}\|_{\textrm{op}}\, dw\xrightarrow{t\rightarrow \infty }0
\end{align*}
as $\psi\perp \Ran(\Omega) $, and by assumption (\ref{EnssCond}), which proves the inclusion.\par
Conversely, let $\psi\in \calH_{\textrm{int}}$. We will show that $\psi \perp \Omega(\calD_k(\mathcal{C}_{i}))$ for any $k>0, i \in \calI$ and conclude by density. Let $\varphi \in \Omega(\calD_k(\mathcal{C}_{i})) $ for some $k>0, i \in \calI$ and let $m,\varepsilon$ and $v$ satisfy $m,v,\varepsilon\in(0,k)$. Then by Lemma  \ref{RanDkLemma} for $\delta$ sufficiently small there exists some $T_0(\calC_i)>0$ such that there are constants $C>0$ and $\ell>0$ so that
\begin{align*}
    \|(P_\delta(\calW_{vt,m;\textrm{out}})-\Id)\varphi_t\|<Ct^{-\ell}+\int\limits_t^\infty\|\chi_{A_{\varepsilon s}}V\|_{\textrm{op}}\|\varphi\|\,ds
\end{align*}
for all $t>T_0$.\par
Then we have that for any $t>T_0$
\begin{align*}
    \braket{\psi,\varphi}&=\braket{P_\delta(\calW_{vt,m;\textrm{out}})\psi_t,\varphi}+\braket{\psi_t,(P_\delta(\calW_{vt,m;\textrm{out}})-\id)\varphi_t}\\
    &\leq\|P_\delta(\calW_{vt,m;\textrm{out}})\psi_t\|\|\varphi\|+\|\psi\|\|(P_\delta(\calW_{vt,m;\textrm{out}})-\Id)\varphi_t\|\\
    &<\|P_\delta(\calW_{vt,m;\textrm{out}})\psi_t\|\|\varphi\|+Ct^{-\ell}+\int\limits_t^\infty\|\chi_{A_{\varepsilon s}}V\|_{\textrm{op}}\|\varphi\|\,ds
\end{align*}
Since  $\psi\in \calH_\textrm{int}$, for the same $v$ and $m$, and choosing $\delta$ smaller if necessary, we have that 
\begin{align*}
    &\lim\limits_{t\rightarrow\infty}\|P_\delta(\calW_{vt,m;\textrm{out}})\psi_t\|=0
\end{align*}
so we may conclude that 
\begin{align*}
    \braket{\psi,\varphi}&<\|P_\delta(\calW_{vt,m;\textrm{out}})\psi_t\|\|\varphi\|+Ct^{-\ell}+\int\limits_t^\infty\|\chi_{A_{\varepsilon s}}V\|_{\textrm{op}}\|\varphi\|\,ds\xrightarrow{t\rightarrow \infty }0
\end{align*}
from assumption (\ref{EnssCond}). Therefore, $\psi \perp \varphi$, as needed.
\end{proof}

\subsection{Spatial characterizations of $\calH_{\textrm{scat}}$ and $\calH_{\textrm{int}}$}
In this section, we show that for some systems one can replace the microlocal descriptions of $\calH_{\textrm{scat}}$ and $\calH_{\textrm{int}}$ with descriptions that are purely spatial. Recall Theorem \ref{spaceOnlyThm}:
\setcounter{section}{3}
\setcounter{theorem}{2}
\begin{theorem}
    Suppose that $\{\calC_i\}_{i\in \calI}$ consists of cones with aperture greater than or equal to $\pi$. Then with $\calD=\overline{\bigcup_{i\in\calI}\calD(C_i)}$ we have that
    \begin{align*}
    &\Omega(\calD)=\overline{\{\psi \in\calH\mid \exists v>0,\lim\limits_{t\rightarrow \infty }\|\chi_{\calA_{vt}^c}\psi_t\|=0\}}\\
    &\Omega(\calD)^\perp=\{\psi \in\calH\mid \forall v>0,\lim\limits_{t\rightarrow \infty }\|\chi_{\calA_{vt}}\psi_t\|=0\}
\end{align*}
\end{theorem}
\setcounter{theorem}{7}
\setcounter{section}{5}
\begin{remark}
The above theorem applies to potentials for which $\{\calC_i\}_{i\in \calI}$ also contains cones of aperture less than $\pi$. In this case, one will have a spatial characterization only for those cones of large enough aperture. See Example \ref{broken} for one such setting.
\end{remark}
So far, we have described the set of scattering states $ \mathcal{H}_{\textrm{scat}} $ as those states which asymptotically propagate into some cone $ \mathcal{C} $ with outgoing momenta, that is, those that point into $ \mathcal{C} $. To obtain a spatial characterization, it suffices to show that it is impossible for a state to propagate into $ \mathcal{C} $ with any other momentum localization, if $\gamma\geq \frac{\pi}{2}$. For this, we begin by defining the incoming subset of phase space for any collection of cones:
let
\begin{align*}
    &W_{n,m;\textrm{in}}(\calC_{x,\vec{v},\gamma})=\{(x,p) \in \bbR^{2d} \mid x\in A_n(\calC_{x,\vec{v},\gamma}), -p \in A_{-m}(\calC_{\vec{v},\gamma})\}\\
    &\calW_{n,m; \textrm{in}}=\bigcup\limits_{i\in\calI}W_{n,m;\textrm{in}}(\calC_i)
\end{align*}
See Figure \ref{fig:momentumInSets}.
\begin{figure}[ht]
\centering
\begin{tikzpicture}
 \pgfmathsetmacro {\coneendx }{3}
 \pgfmathsetmacro {\coneendy }{4}
 \pgfmathsetmacro {\smallconeendx }{sqrt(3/7)}
 \pgfmathsetmacro {\smallconeendy }{4/3*\smallconeendx }
 \pgfmathsetmacro {\pnty }{4.5}
 \pgfmathsetmacro {\pntx }{0.2}
 \pgfmathsetmacro {\An }{2}
 \pgfmathsetmacro {\Ak }{0.4}
 \pgfmathsetmacro {\deg }{37}
\draw[orange,dashed] (\coneendx,\coneendy) -- (0,0) -- (-\coneendx,\coneendy);     
\filldraw[orange]  (0,0) circle (0pt) node[anchor=north] {$\calC_i$};

\draw[red,dashed,thick] (0,0) -- (0,1);
\draw[red] (0,1) arc (90:60:1);
\filldraw[red] (-0.1,0.7) circle (0pt) node[anchor=west] {$\gamma$};

\draw[black,thick] (\coneendx,\coneendy+\An) -- (0,0+\An) -- (-\coneendx,\coneendy+\An);      
\filldraw[black]  (0,0+\An) circle (0pt) node[anchor=north] {$A_n(\calC_i)$};

\draw (\pntx,\pnty) node {x};
\draw[red,pattern=north east lines, pattern color=red] (\pntx,\pnty+\Ak) -- (\pntx-\smallconeendx,\pnty+\Ak-\smallconeendy) arc(-90-\deg:-90+\deg:1) -- cycle;
\filldraw[red] (\pntx,\pnty-1+\Ak) circle (0pt) node[anchor=north] {$W_{n,m;\textrm{in}}$};

\draw[blue,pattern=north east lines, pattern color=blue] (\pntx,\pnty+\Ak) -- (\pntx-\smallconeendx,\pnty+\smallconeendy+\Ak) arc(90+\deg:90-\deg:1) -- cycle;
\filldraw[blue] (\pntx-\smallconeendx,\pnty+\smallconeendy+\Ak) circle (0pt) node[anchor=east] {$W_{n,m;\textrm{out}}$};

\end{tikzpicture}
\caption{Illustration of the phase space sets $W_{n,m;\textrm{out}}(\calC_i)$ and $W_{n,m;\textrm{in}}(\calC_i)$: each has space coordinates inside the black cone with momentum coordinates inside the red/blue cone, respectively.}\label{fig:momentumInSets}
\end{figure}
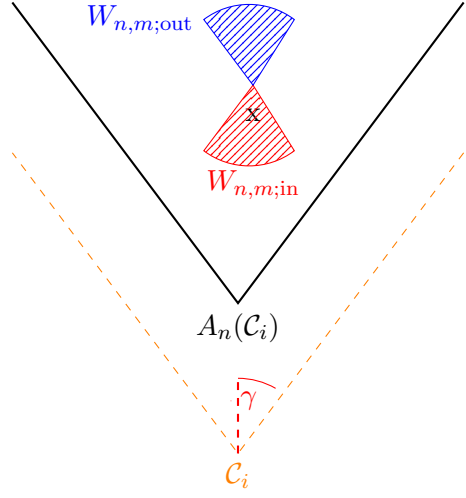
We show that asymptotically no state can concentrate in these subsets of phase space:
\begin{proposition}\label{InProp}
For any $v>0$, $0<m<v$, and $\delta< \frac{v-m}{2}$
\begin{align}\label{inLim}
    & \slim\limits_{t\rightarrow \infty} P_\delta(\calW_{vt,m;\mathrm{in}})e^{-itH}=0
\end{align}
\end{proposition}
\begin{proof}
This proof is based on an argument of Enss recorded in \cite{SimonEnns}. For any $\psi \in \calH$ we can write 
\begin{align*}
     &\|P_{\delta }(\calW_{vt,m; \textrm{in}})e^{-itH}\psi\|\leq\|P_{\delta }(\calW_{vt,m; \textrm{in}})(e^{-itH}-e^{-itH_0})\psi\|+\|P_{\delta }(\calW_{vt,m; \textrm{in}})e^{-itH_0}\psi\| 
\end{align*}
so to prove (\ref{inLim}), it suffices to prove that for $v,m,$ and $\delta$ as above
\begin{align}
    &\lim_{t\rightarrow\infty}\|P_{\delta }(\calW_{vt,m; \textrm{in}})(e^{-itH}-e^{-itH_0})\|_{\textrm{op}}=0\label{eq:2}
\end{align}
and
\begin{align}
   & \slim\limits_{t\rightarrow \infty}P_{\delta }(\calW_{vt,m; \textrm{in}})e^{-itH_0}=0\label{eq:3}
\end{align}
To prove (\ref{eq:2}), we write, for $\varepsilon<\frac{v-m}{4}$
\begin{align*}
    &\|P_{\delta }(\calW_{vt,m; \textrm{in}})(e^{-itH}-e^{-itH_0})\|_\textrm{op}= \|(e^{itH}-e^{itH_0})P_{\delta }(\calW_{vt,m; \textrm{in}})\|_\textrm{op}\\
    &=\|(\Id -e^{-itH}e^{itH_0})P_{\delta }(\calW_{vt,m; \textrm{in}})\|_\textrm{op}\leq \int\limits_{0}^t\|e^{-iwH}(-H+H_0) e^{iwH_0}P_\delta(\calW_{vt,m; \textrm{in}})\|_{\textrm{op}}\,dw\\
    &\leq \int\limits_{0}^t\|\chi_{\calA_{\varepsilon (t+w)}}V\|_{\textrm{op}}\,dw  +M\int\limits_{0}^t\|\chi_{\calA_{\varepsilon (t+w)}^c} e^{iwH_0}P_\delta(\calW_{vt,m; \textrm{in}})\|_{\textrm{op}}\,dw\\
    &\leq  \int\limits_{t}^\infty \|\chi_{\calA_{\varepsilon w }}V\|_{\textrm{op}}\,dw  +M\int\limits_{0}^t\|\chi_{\calA_{\varepsilon (t+w)}^c} e^{iwH_0}P_\delta(\calW_{vt,m; \textrm{in}})\|_{\textrm{op}}\,dw
\end{align*}
Now we note that for any cone $\calC_{x,\vec{v},\gamma}$
\begin{align*}
     \mathfrak{C}_{-w}(W_{vt,m; \textrm{in}}(\calC_{x,\vec{v},\gamma}))&=\{x-wp\mid (x,p)\in W_{vt,m; \textrm{in}}(\calC_{x,\vec{v},\gamma})\}\\
     &=\{x+wp\mid (x,p)\in A_{vt}(\calC_{x,\vec{v},\gamma}), p\in A_{-m}(\calC_{\vec{v},\gamma})\}\\
     &=\mathfrak{C}_{w}(W_{vt,-m;\textrm{out}})
\end{align*}
so that by Proposition \ref{coneOnlyGeo}
\begin{align*}
    d(\mathfrak{C}_{-w}(\calW_{vt,m; \textrm{in}}),\calA_{\varepsilon (t+w)}^c)\geq (vt-mw)-\varepsilon(t+w)
\end{align*}
which is greater than $\delta w$ because
$w<t$, $\delta<\frac{v-m}{2}$,and $\varepsilon<\frac{v-m}{4}$.
Thus, we may apply Lemma \ref{generalnonstationary} to conclude that for any $\ell>0$ there is some $C>0$ such that
\begin{align*}
    \|\chi_{\calA_{\varepsilon (t+w)}^c} e^{iwH_0}P_\delta(\calW_{vt,m; \textrm{in}})\|_{\textrm{op}}<Ct^{-\ell}
\end{align*}
from which (\ref{eq:2}) follows immediately when combined with the Enss condition (\ref{EnssCond}).\par
To prove (\ref{eq:3}), we fix $\psi \in \calH$ compactly supported and choose $R$ so that $\supp \psi \subset \calA_0+B_R$. Then
\begin{align*}
    &\|P_{\delta }(\calW_{vt,m;\textrm{in}})e^{-iH_0t}\psi \|=\|P_{\delta }(\calW_{vt,m;\textrm{in}})e^{-iH_0t}\chi_{\calA_0+B_R}\psi \|\leq \|\chi_{\calA_0+B_R} e^{iH_0t}P_{\delta }(\calW_{vt,m;\textrm{in}})\|_\textrm{op} \|\psi\|
\end{align*}
Again by Proposition \ref{coneOnlyGeo}
\begin{align*}
    d(\mathfrak{C}_{-t}(\calW_{vt,m; \textrm{in}}),\calA_0+B_R)>(v-m)t-R
\end{align*}
and thus for $t>\frac{2R}{v-m}$
\begin{align*}
     d(\calA_0+B_R,\mathfrak{C}_{-t}(\calW_{vt,m; \textrm{in}}))>\frac{v-m}{2}t>\delta t
\end{align*}
Therefore, we can apply Lemma \ref{generalnonstationary}, to get that
\begin{align*}
    \|\chi_{\calA_0+B_R} e^{iH_0t}P_{\delta }(W_{vt,m;\textrm{in}})\|_\textrm{op}< C((v-m)t-R)^{-\ell}
\end{align*}
from which it follows that
\begin{align*}
    \lim_{t\rightarrow\infty}P_\delta(\calW_{vt,m;\textrm{in}})e^{-itH_0}\psi=0
\end{align*}
Density establishes (\ref{eq:3}), thus proving the lemma in full.
\end{proof}
\begin{proof}[Proof of Theorem \ref{spaceOnlyThm}]
The key point is that in this case
\begin{align}\label{inOutGeo}
    \calW_{n,m;\textrm{out}}\cup \calW_{n,m;\textrm{in}}=\calA_n\times \bbR^d
\end{align}
To see this, note that if $\calC_i=\calC_{x,\vec{v},\gamma}$ and $\gamma\geq \frac{\pi}{2}$ (that is, the aperture of the cone is at least $\pi$) then $\calC_{\vec{v},\gamma}^c\subset -\calC_{\vec{v},\gamma}$ since if $y\in \calC_{\vec{v},\gamma}^c$ we have
\begin{align*}
     &\braket{y,\vec{v}}\leq \cos(\gamma)\|y\| \implies\braket{-y,\vec{v}}\geq -\cos(\gamma)\|y\|>\cos(\gamma)\|y\|
\end{align*}
as $\cos(\gamma)<0$ if $\gamma>\frac{\pi}{2}$ and if $\gamma=\frac{\pi}{2}$ the above is true up to a set of measure $0$. In particular,
\begin{align*}
    A_m^c(\calC_{\vec{v},\gamma})=\calC_{\vec{v},\gamma}^c+\frac{m}{\sin(\gamma)}\vec{v}\subset-\calC_{\vec{v},\gamma}+\frac{m}{\sin(\gamma)}\vec{v}=-A_{-m}(\calC_{\vec{v},\gamma})
\end{align*}
so that (\ref{inOutGeo}) holds.\par
Now, fix $\psi\in\calH_\textrm{int}$ and $v>0$. Choose $m<v$ and $\delta$ sufficiently small and apply Proposition \ref{InProp} to see that
\begin{align*}
    \lim\limits_{t\rightarrow \infty }\|P_\delta(\calA_{vt}\times \bbR^d)\psi_t\|=\lim\limits_{t\rightarrow \infty }\|P_\delta(\calW_{vt,m;\mathrm{out}}\cup \calW_{vt,m;\mathrm{in}})\psi_t\|=0
\end{align*}
Since
\begin{align*}
    \|\chi_{\calA_{\frac{vt}{2}}}\psi_t\|\leq \|P_\delta(\calA_{vt})\psi_t\|+\|\chi_{\calA_{\frac{vt}{2}}}P_\delta(\calA_{vt}^c)\psi_t\|
\end{align*}
and
\begin{align*}
    \|P_\delta(\calA_{vt})\psi_t\|\leq \|\chi_{\calA_{2vt}}\psi_t\|+\|P_\delta(\calA_{vt})\chi_{\calA_{2vt}^2}\psi_t\|
\end{align*}
from Proposition \ref{SpaceLoc} we see that
\begin{align*}
    \|\chi_{\calA_{\frac{vt}{2}}}\psi_t\| + o(1)\leq \|P_\delta(\calA_{vt})\psi_t\|\leq \|\chi_{\calA_{2vt}}\psi_t\| + o(1)
\end{align*}
as $t\rightarrow \infty$. Therefore,
\begin{align*}
    \calH_\textrm{int}\subset \{\psi\in\calH\mid \forall v>0\lim_{t\rightarrow0}\|\chi_{\calA_{vt}}\psi_t\|=0\}
\end{align*}
Conversely, if $\lim\limits_{t\rightarrow\infty}\|\chi_{\calA_{vt}}\psi_t\|=0$ then by the above for any $\delta>0$
\begin{align*}
    \lim_{t\rightarrow\infty}\|P_\delta(\calA_{vt})\psi_t\|=0
\end{align*}
and thus for any $m>0$
\begin{align*}
    \|P_\delta(\calW_{vt,m;\textrm{out}})\psi_t\|^2= \Span{P_\delta(\calW_{vt,m;\textrm{out}})^2\psi_t,\psi_t}\leq \braket{P_\delta(\calA_{vt}\times\bbR^d)\psi_t,\psi_t}\xrightarrow{t\rightarrow\infty}0
\end{align*}
This proves the opposite inclusion
\begin{align*}
    \{\psi\in\calH\mid \forall v>0\lim_{t\rightarrow0}\|\chi_{\calA_{vt}}\psi_t\|=0\}\subset \calH_\textrm{int}
\end{align*}
The same argument shows that
\begin{align*}
    \calH_\textrm{scat}\subset\{\psi\in\calH\mid \exists v>0\lim_{t\rightarrow0}\|\chi_{\calA_{vt}^c}\psi_t\|=0\}
\end{align*}
because $\sloppy{\calA_{vt}^c\times\bbR^d\subset \calW_{vt,m;\textrm{out}}^c}$ for any $m>0$. Furthermore, if $\lim\limits_{t\rightarrow\infty}\|\chi_{\calA_{vt}^c}\psi_t\|=0$ for some $v>0$, then $\psi$ is orthogonal to $\calH_\textrm{int}$ since we have shown that any $\varphi\in\calH_\textrm{int}$ must satisfy $\lim\limits_{t\rightarrow\infty}\|\chi_{\calA_{vt}}\varphi_t\|=0$, for all $v>0$. Therefore, $\psi \in \calH_\textrm{int}^\perp=\overline{\calH_{\textrm{scat}}}$, thus proving the opposite inclusion.
\end{proof}

\section{Examples}
\setcounter{section}{2}
\begin{example}[Single cone]
Suppose that $\{\calC\}_{i\in\calI}$ consists of a single cone $\calC=\calC_{x,\vec{v},\gamma}$. Then $\calD=\calC_{\vec{v},\gamma}$ and Theorem \ref{thms} gives the following microlocal description:
    \begin{align*}
        &\Omega(\calD)=\overline{\{\psi \in \calH\mid \exists v,m,\delta_0>0,\text{ so that  }\forall \delta\in(0,\delta_0)\lim_{t\rightarrow\infty}\|(P_\delta([A_{vt}(\calC)\times A_m(\calC)]^c)\psi_t\|=0\}}\\
        &\Omega(\calD)^\perp=\{\psi \in \calH \mid  \forall v,m>0, \exists \delta_0>0 \text{ so that  }\forall \delta\in(0,\delta_0)\,\,\lim\limits_{t\rightarrow\infty} \|P_\delta(A_{vt}(\calC)\times A_m(\calC))\psi_t\|=0\}
    \end{align*}
    This indicates that $\Omega(\calD)$ consists of states which propagate into $\calC$ with momenta in $A_m(\calC)$. When $\gamma <\frac{\pi}{2}$, this is the best description our theorems afford. It does not rule out a state in $\Omega(\calD)^\perp$ which propagates into $\calC$, but with the wrong momenta and that thus could ``bounce" off of the boundary of $\calC$.\par 
    However, when $\gamma\geq \frac{\pi}{2}$, Theorem \ref{spaceOnlyThm} shows that in fact
    \begin{align*}
        &\Omega(\calD)=\overline{\{\psi \in\calH\mid \exists v>0,\lim\limits_{t\rightarrow \infty }\|\chi_{A_{vt}^c(\calC)}\psi_t\|=0\}}\\
        &\Omega(\calD)^\perp=\{\psi  \in \calH\mid  \forall v>0\,\,\lim\limits_{t\rightarrow\infty} \|\chi_{A_{vt}(\calC)}\psi_t\|=0\}
    \end{align*}
    because in this case we have shown that is impossible for a state to propagate into $ \calC $ with momenta pointing away from $\mathcal{C}$ (this is the content of Proposition \ref{InProp}).
\end{example}
\begin{example}[Short-range scattering]
	As explained in the introduction, we may choose $ \{\mathcal{C}\}_{i\in \calI} $ so that $ \calA_r=B_r^c $. Relative to this collection of cones, the condition (\ref{EnssCond}) becomes the classical Enns condition
	\begin{align*}
		\|V\chi_{B_{r}^c}\|_\textrm{op}\in L^1([0,\infty),dr)
	\end{align*}
	which is one of many short-range scattering assumptions in the literature. Here, $\calD$ is in fact equal to all of $\mathcal{H}$. In this setting, Theorem \ref{spaceOnlyThm} shows that
    \begin{align*}
	&\Ran(\Omega)=\overline{\{\psi \in\calH\mid \exists v>0,\lim\limits_{t\rightarrow \infty }\|\chi_{B_{vt}}\psi_t\|=0\}}\\
        &\Ran(\Omega)^\perp=\{\psi \in\calH\mid \forall v>0,\lim\limits_{t\rightarrow \infty }\|\chi_{B_{vt}^c}\psi_t\|=0\}
        \end{align*}
    This result may be contrasted with the usual asymptotic completeness statement for short-range scattering, which is
    \begin{align*}
    	&\Ran(\Omega)=\mathcal{H}_{\textrm{c}}(H)\\
    	&\Ran(\Omega)^\perp=\mathcal{H}_{\textrm{pp}}(H)
    \end{align*}
    This latter description may be connected to the dynamics of $H$ via the RAGE theorem \cite{AmreinGeorgescu,Ruelle}, which is a crucial ingredient in the original argument of Enss. A standard formulation of the RAGE theorem (see for example \cite{teschl2009mathematical}) is
    \begin{align*}
    	&\calH_{\textrm{c}}(H)=\{\psi\in \calH\mid \lim\limits_{n\rightarrow\infty }\lim\limits_{T\rightarrow\infty } \frac{1}{T}\int\limits_0^T\|\chi_{B_n}\psi_t\|dt=0\}\\
        &\calH_{\textrm{pp}}(H)=\{\psi\in \calH\mid \lim\limits_{n\rightarrow\infty }\sup\limits_{t\geq 0} \|\chi_{B_n^c}\psi_t\|dt=0\}
    \end{align*}
    In particular, the space varible $n$ is decoupled from $t$, whereas in our result $ n=vt $ for some velocity $v$. 
\end{example}
\begin{example}[Subspace potential]
	Let $S_r$ be the points within distance $ r $ from some fixed subspace of $ \bbR^d $. We explained in the introduction that $ S_r $ may be written as $\calA_r^c$ for some appropriately chosen collection of cones. In this setting, $ \calD=\mathcal{H} $ and Theorem \ref{spaceOnlyThm} shows that
	if
	\begin{align*}
		\|V\chi_{S_{r}^c}\|_\textrm{op}\in L^1([0,\infty),dr)
	\end{align*}
	then
    \begin{align*}
        &\Ran(\Omega)=\overline{\{\psi \in\calH\mid \exists v>0,\lim\limits_{t\rightarrow \infty }\|\chi_{S_{vt}}\psi_t\|=0\}}\\
        &\Ran(\Omega)^\perp=\{\psi \in\calH\mid \forall v>0,\lim\limits_{t\rightarrow \infty }\|\chi_{S_{vt}^c}\psi_t\|=0\}
    \end{align*}
    which recovers the main result of \cite{black2022scattering}.
\end{example}
\begin{example}[Broken subspace]
	In this case, $ \{\mathcal{C}\}_{i\in\calI} $ consists of two cones $ \mathcal{C}_1 $ and $ \mathcal{C}_2 $, the first with $ \gamma_1<\frac{\pi}{2} $ and the second with $ \gamma_2=\pi-\frac{\pi}{2} $. Let $ \calD_1=\calD(\calC_1) $ and $ \calD_2=\calD(\calC_2) $ be the domains of $ \Omega $ corresponding to each cone. Then relative to $ \mathcal{C}_1 $ we obtain only a microlocal description
     \begin{align*}
        &\Omega(\calD_1)=\overline{\{\psi \in \calH\mid \exists v,m,\delta_0>0\text{ so that }\forall\delta\in (0,\delta_0) \lim_{t\to\infty}\|(P_\delta([A_{vt}(\calC_1)\times \calC_1]^c)\psi_t\|=0\}}\\
        &\Omega(\calD_1)^\perp=\{\psi  \in \calH\mid  \forall v,m>0, \exists \delta_0>0, \text{ so that  }\forall \delta\in(0,\delta_0)\,\,\lim\limits_{t\rightarrow\infty} \|P_\delta(A_{vt}(\calC_1)\times A_m(\calC_1))\psi_t\|=0\}
    \end{align*}
    whereas for the second cone we obtain a purely spatial description
    \begin{align*}
         &\Omega(\calD_2)=\overline{\{\psi \in\calH\mid \exists v>0,\lim\limits_{t\rightarrow \infty }\|\chi_{A_{vt}(\calC_2)}\psi_t\|=0\}}\\
        &\Omega(\calD_2)^\perp=\{\psi \in\calH\mid \forall v>0,\lim\limits_{t\rightarrow \infty }\|\chi_{A_{vt}^c(\calC_2)}\psi_t\|=0\}
    \end{align*}
    In other words, any state which propagates into the larger cone $ \mathcal{C}_2 $ at a linear rate must be a scattering state, but for $ \mathcal{C}_1 $ this is only the case for states with momenta which also point into $ \mathcal{C}_1 $.
\end{example}
\setcounter{section}{6}
\appendix
\section{Existence of the POVM $P_\delta$}\label{DaviesProperties}
\begin{lemma}
There exists a Positive Operator Valued Measure (POVM), $P_\delta$, defined on the phase space $\bbR^d_x\times \bbR^d_p$, with the following properties, for any $E\subset \bbR^{2d}$ Borel
\begin{enumerate}
    \item(Observable) $P_\delta(\bbR^{2d})=\id$.
    \item(Momentum localization) Let $B\subset \bbR^d$ and $D\subset \bbR^d$ be Borel sets such that $d(B,D)>\delta$. Then for any $E\subset \bbR^d\times B$ Borel and $\psi \in \calH$ such that $\supp \hat{\psi} \subset D$
    \begin{align*}
        P_\delta(E) \psi =0
    \end{align*}
    \item(Approximate space localization) Let $A\subset \bbR^d$ and $D\subset\bbR^d$ be Borel sets so that $\sloppy{d(D,A)>0}$. Then for any $\ell>0$ there exists some constant $C>0$ so that for all $E\subset A\times \bbR^d$
    \begin{align*}
        \|P_{\delta}(E)\chi_{D}\|_{\textrm{op}}<C[d(A,D)]^{-\ell}
    \end{align*}
    \item(Microlocal non-stationary phase estimate) Let $\mathfrak{C}_t(E)\subset \bbR^{d}$ denote the classically allowed region associated to $E\subset \bbR^{2d}$ at time $t$:
    \begin{align*}
        \mathfrak{C}_t(E)=\{x+tp\mid (x,p)\in E\}
    \end{align*}
    Let $F\subset \bbR^d$ be Borel. For any $\ell>0$ there exists $C>0$ such that 
    \begin{align*}
        \|\chi_{F}e^{-itH_0}P_\delta(E)\|_\mathrm{op}\leq Cd(|t|)^{-\ell}
    \end{align*}
    for all $t$ such that $d(t):=d(\mathfrak{C}_t(E),F)>\delta |t|$.
    \item(Spatial non-stationary phase estimate) Let $\{A_t\}_{t\geq0}$ be collections of Borel subsets of $\bbR^d$.\par
    Then for any $\varphi \in \calS$ such that $\supp \hat{\varphi} \Subset D$ Borel, $\ell>0$, and $\varepsilon>0$ there exists some constant $C(\psi,\ell,\varepsilon,\delta)>0$ such that 
    \begin{align*}
        \|P_{\delta}(A_t\times\bbR^d)e^{-itH_0}\varphi \|< Ct^{-\ell}
    \end{align*}
    for all $t$ such that $d(A_{t},tD)>\varepsilon t$.
\end{enumerate}
\end{lemma}
\begin{proof}
To this end, we will use the phase space observable formalism developed in \cite{davies1976quantum,davies1980enss} and used in \cite{black2022scattering}.\par
Choose $\eta\in\calS$, such that $\|\eta\|=1$ and $\supp \hat{\eta}\subset B_1$. Let $\eta_\delta$ be such that $\hat{\eta}_\delta(p)=\delta^{-\frac{d}{2}}\hat{\eta}(\frac{p}{\delta})$, a rescaling of $\eta$, so that $\supp \hat{\eta}_\delta\subset B_\delta$ and $\|\eta_{\delta}\|=1$.\par
Now define the following family of coherent states by translating $\eta_\delta$ in phase space:
\begin{align*}
    &\hat{\eta}_{x,p;\delta}(\xi)=e^{-ix\xi}\hat{\eta}_\delta(\xi-p)
\end{align*}
We use this to define a family, depending on $\delta>0$, of positive-operator-valued measures: for any $E\subset \bbR^{2d}$ Borel and $\psi \in \calH$ let
\begin{align*}
    P_\delta(E)\psi=(2\pi)^{-d}\iint\limits_E \braket{\eta_{x,p;\delta},\psi} \eta_{x,p;\delta} \,dx\,dp
\end{align*}
The various properties of $P_\delta$ are proved in a series of propositions below.\par
In Appendix A of \cite{black2022scattering} we proved the following properties of $P_\delta$:
\begin{proposition}[Observable]
For any $\delta>0$ we have $ P_\delta(\bbR^{2d})=\id$.
\end{proposition}
\begin{proposition}[Momentum localization]\label{psupport}
Let $B\subset \bbR^d$ and $D\subset \bbR^d$ be Borel sets such that $d(B,D)>\delta$. Then for any $E\subset \bbR^d\times B$ Borel and $\psi \in \calH$ such that $\supp \hat{\psi} \subset D$
    \begin{align*}
        P_\delta(E) \psi =0
    \end{align*}
\end{proposition}
\begin{proposition}[Approximate space localization]\label{SpaceLoc}
 Let $A\subset \bbR^d$ Borel and any set $D\subset\bbR^d$ Borel such that $d(D,A)>0$, for any $\ell>0$ we have some constant $C>0$ so that
    \begin{align*}
        \|P_{\delta}(A\times \bbR^d)\chi_{D}\|_{\mathrm{op}}<C[d(A,D)]^{-\ell}
    \end{align*}
\end{proposition}

Finally we prove two estimates relating $P_\delta$ to the free propagator $e^{-itH_0}$, both based on the principle of of non-stationary phase. The first is similar to Lemma 2 of Theorem XI.112 in \cite{RSVol3}, but adapted to $P_\delta$. This lemma and its proof are similar to Lemma 3 in \cite{yoneyama}. 
\begin{lemma}[Microlocal non-stationary phase estimate]\label{generalnonstationary}
Let $\mathfrak{C}_t(E)\subset \bbR^{d}$ denote the classically allowed region associated to $E\subset \bbR^{2d}$ at time $t$:
    \begin{align*}
        \mathfrak{C}_t(E)=\{x+tp\mid (x,p)\in E\}
    \end{align*}
    Let $F\subset \bbR^d$ be Borel. For any $\ell>0$ there exists $C>0$ such that 
    \begin{align*}
        \|\chi_{F}e^{-itH_0}P_\delta(E)\|_\mathrm{op}\leq Cd(|t|)^{-\ell}
    \end{align*}
    for all $t$ such that $d(t):=d(\mathfrak{C}_t(E),F)>\delta |t|$.
\end{lemma}
\begin{proof}
We start by noting that for any $\psi \in \calH$, by the boundedness of $P_\delta$
\begin{align*}
     \|P(E)\psi\|^2=\braket{\psi,P^2(E)\psi}&\leq \braket{\psi,P(E)\psi}=(2\pi)^{-d}\iint\limits_E \braket{\eta,\psi}\braket{\psi,\eta} dx dp=(2\pi)^{-d}\iint\limits_E |\braket{\eta,\psi}|^2dx dp
\end{align*}
We will estimate the norm of the adjoint operator $P_\delta(E)e^{itH_0}\chi_F$. For $\psi \in\calH$, by the above inequality 
\begin{align*}
    \|P_\delta(E)e^{itH_0}\chi_F\psi\|^2&\leq (2\pi)^{-d}\iint\limits_E|\Span{\eta_{x,p;\delta},e^{itH_0}\chi_F\psi}|^2\,dx\,dp\\
    &=(2\pi)^{-d}\iint\limits_E|\int
    \limits_{\bbR^d}\overline{e^{-itH_0}\eta_{x,p;\delta}}(y)\chi_F(y)\psi(y)\,dy|^2\,dx\,dp
\end{align*}
We now compute
\begin{align*}
    &(e^{-itH_0}\eta_{x,p;\delta})(y)=(2\pi)^{-\frac{d}{2}}\int\limits_{\bbR^d}e^{i\xi\cdot(y-x)-it\frac{\xi^2}{2}}\hat{\eta}_\delta(\xi-p)\,d\xi\\
    &=e^{ip\cdot(y-x)-it\frac{p^2}{2}}(2\pi)^{-\frac{d}{2}}\int\limits_{\bbR^d}e^{i\xi\cdot(y-x)-it\frac{\xi^2}{2}-it\xi\cdot p}\hat{\eta}_\delta(\xi)\,d\xi=e^{ip\cdot(y-x)-it\frac{p^2}{2}}(e^{-itH_0}\eta_\delta)(y-(x+tp))
\end{align*}
Recalling that for any $(x,p)\in E$ and $y\in F$, $|y-(x+tp)|>d(t)
$, we see that
\begin{align*}
    &|\int\limits_{\bbR^d}\overline{e^{-itH_0}\eta_{x,p;\delta}}(y)\chi_F(y)\psi(y)\,dy|=|\int\limits_{\bbR^d}e^{-ipy}\overline{(e^{-itH_0}\eta_\delta)}(y-(x+tp))\chi_F(y)\psi(y)\,dy|\\
    &=|\int\limits_{\bbR^d}e^{-ipy}\overline{(e^{-itH_0}\eta_\delta)}(y-(x+tp))\chi_{\{|y-(x+tp)|>d(t)\}}(y)\chi_F(y)\psi(y)\,dy|\\
    &=(2\pi)^{\frac{d}{2}}|\calF[\overline{(e^{-itH_0}\eta_\delta)}(\cdot-(x+tp))\chi_{\{|\cdot-(x+tp)|>d(t)}\}(\cdot)\chi_F(\cdot)\psi(\cdot)](p)|
\end{align*}
We now perform the change of variables $x'=x+tp$ and apply the Plancharel theorem to see that
\begin{align*}
    &(2\pi)^{-d}\iint\limits_E|\int
    \limits_{\bbR^d}\overline{e^{-itH_0}\eta_{x,p;\delta}}(y)\chi_F(y)\psi(y)\,dy|^2\,dx\,dp\\
    &\leq(2\pi)^{-d} \iint\limits_{\bbR^{2d}}(2\pi)^d|\calF[\overline{(e^{-itH_0}\eta_\delta)}(\cdot-(x+tp))\chi_{\{|\cdot-(x+tp)|>d(t)\}}(\cdot)\chi_F(\cdot)\psi(\cdot)](p)|^2\,dx\,dp\\
    &=_{x'=x+tp}\iint\limits_{\bbR^{2d}}|\calF[\overline{(e^{-itH_0}\eta_\delta)}(\cdot-x')\chi_{\{|\cdot-x'|>d(t)\}}(\cdot)\chi_F(\cdot)\psi(\cdot)](p)|^2\,dx'\,dp\\
    &=\iint\limits_{\bbR^{2d}}|\overline{(e^{-itH_0}\eta_\delta)}(y-x')\chi_{\{|y-x'|>d(t)\}}(y)\chi_F(y)\psi(y)|^2\,dx'\,dy\\
    &=\int\limits_{\bbR^{d}}|\chi_F(y)\psi(y)|^2\int\limits_{\{x'\mid |y-x'|>d(t)\}}|\overline{(e^{-itH_0}\eta_\delta})(y-x')|^2\,dx'\,dy\leq \|\psi\|^2\int\limits_{B_{d(t)}^c}|(e^{-itH_0}\eta_\delta)(x')|^2\,dx'
\end{align*}
Since $d(t)>\delta t$ by assumption and $\supp \hat{\eta}\subset B_\delta$, we see that if $x'\in B_{d(t)}^c$ then $\frac{x'}{t}\not\in B_\delta$ so we may apply Lemma \ref{simpleNonstationary}
to see that for any $\ell>0$ there exists $C>0$ depending only on $\eta$ and $\delta$ such that
\begin{align*}
    \int\limits_{\{x'\mid |y-x'|>d(t)\}}|\overline{(e^{-itH_0}\eta_\delta})(y-x')|^2\,dx'\leq C\int\limits_{B_{d(t)}^c}(1+\|x'\|+|t|)^{-\ell}\,dx'\leq C(1+d(t)+|t|)^{-\ell+d}
\end{align*}
Thus, we conclude that
\begin{align*}
    \|P_\delta(E)e^{itH_0}\chi_F\psi\|\leq C(1+|t|+d(t))^{-\ell+d}\|\psi\|^2
\end{align*}
as claimed.
\end{proof}
The second lemma is essentially a quite standard non-stationary phase estimate on $e^{-itH_0}$, see for instance the Corollary to Theorem XI.14 from \cite{RSVol3}.
\begin{proposition}[Spatial non-stationary phase estimate]\label{BasicNonstationaryPd}
 Let $\{A_t\}_{t\geq0}$ be a collection of Borel subsets of $\bbR^d$. Then for any $\varphi \in \calS$ such that $\supp \hat{\varphi} \Subset D$ Borel, $\ell>0$, and $\varepsilon>0$ there exists some constant $C(\psi,\ell,\varepsilon,\delta)>0$ such that 
    \begin{align*}
        \|P_{\delta}(A_t\times\bbR^d)e^{-itH_0}\varphi \|< Ct^{-\ell}
    \end{align*}
    for all $t$ such that $d(A_{t},tD)>\varepsilon t$.
\end{proposition}
\begin{proof}
Let $\varphi\in\calS$ and $\ell,\varepsilon>0$ and $D$ be as above. Then we can write
\begin{align*}
     \|P_{\delta}(A_{t} \times \bbR^d)e^{-itH_0}\varphi\|\leq \|P_{\delta}(A_{t} \times \bbR^d)\chi_{[A_{t}+B_{\frac{\varepsilon}{2}t}]^c}\|_{\textrm{op}} \|\varphi\|+\|\chi_{A_{t}+B_{\frac{\varepsilon}{2}t}}e^{-itH_0}\varphi\|
\end{align*}
Since $d(A_t,[A_{t}+B_{\frac{\varepsilon}{2}t}]^c)>\frac{\varepsilon}{2}t$, by Property \ref{SpaceLoc} we get that
\begin{align*}
    \|P_{\delta}(A_{t} \times \bbR^d)\chi_{[A_{t}+B_{\frac{\varepsilon}{2}t}]^c}\|_{\textrm{op}}<Ct^{-\ell}
\end{align*}
We can write
\begin{align*}
    \|\chi_{A_{t}+B_{\frac{\varepsilon}{2}}}e^{-itH_0}\varphi\|^2=\int\limits_{A_{t}+B_{\frac{\varepsilon}{2}}} |e^{-itH_0}\varphi(y)|^2 dy
\end{align*}
Next we note that 
\begin{align*}
    d(A_t+B_{\frac{\varepsilon}{2}t},tD)\geq d(A_t,tD)-\frac{\varepsilon}{2}t>\frac{\varepsilon}{2}t
\end{align*}
So we conclude that $y\in A_t+B_{\frac{\varepsilon}{2}t} $ implies that  $\frac{y}{t}\not \in D$, and so by Lemma \ref{classicnonstationary} we get 
\begin{align*}
     |e^{-itH_0}\varphi(y)|\leq C(1+\|y\|+t)^{-\ell-d}
\end{align*}
Therefore,
\begin{align*}
    \|\chi_{A_t+B_{\frac{\varepsilon}{2}t}}e^{-itH_0}\varphi\|^2\leq (1+t)^{-\ell}
\end{align*}
as needed.
\end{proof}

With this lemma we proved all the needed properties for $P_\delta$.
\end{proof}

\bibliographystyle{amsplain}
\bibliography{bibliography,bibliography2}
\end{document}